\documentclass[10pt,twocolumn,superscriptaddress,aps,prx,balancelastpage]{revtex4-2}
\usepackage[utf8]{inputenc}
\usepackage{amsmath,amsfonts,amssymb,graphicx,physics,mathtools,braket}
\usepackage[dvipsnames]{color,xcolor,colortbl}
\usepackage{hyperref}
\hypersetup{
  breaklinks=true,
  colorlinks=true,
  linkcolor=NavyBlue,
  citecolor=NavyBlue,
  urlcolor=NavyBlue,
  plainpages=false
}
\usepackage{CJKutf8}

\definecolor{RED4}{rgb}{0.75,0,0}

\newtheorem{theorem}{Theorem}

\makeatletter
\let\latex@addcontentsline\addcontentsline
\renewcommand{\addcontentsline}[3]{}
\makeatother

\begin{document}

\title{Entanglement advantage in sensing power-law spatiotemporal noise correlations}

\author{\begin{CJK}{UTF8}{gbsn}Yu-Xin Wang (王语馨)\end{CJK}}
\email{yxwang.physics@outlook.com}
\affiliation{Joint Center for Quantum Information and Computer Science, NIST/University of Maryland, College Park, MD, 20742, USA}

\author{Anthony J. Brady}
\affiliation{Joint Center for Quantum Information and Computer Science, NIST/University of Maryland, College Park, MD, 20742, USA}
\affiliation{Joint Quantum Institute, NIST/University of Maryland, College Park, MD, 20742, USA}

\author{Federico Belliardo}
\affiliation{Pritzker School of Molecular Engineering, University of Chicago, Chicago, IL 60637}

\author{{Alexey~V.~Gorshkov}}
\affiliation{Joint Center for Quantum Information and Computer Science, NIST/University of Maryland, College Park, MD, 20742, USA}
\affiliation{Joint Quantum Institute, NIST/University of Maryland, College Park, MD, 20742, USA}

\begin{abstract} 
Noise sensing underlies many physical applications including tests of non-classicality, thermometry, verification of correlated phases of quantum matter, and characterization of criticality. While previous works have shown that quantum resources such as entanglement and squeezing can enhance the sensitivity in estimating deterministic signals, less is known about the entanglement advantage in sensing correlated stochastic signals (noise). In this work, we compute the fundamental sensitivity limits of quantum sensors in probing spatiotemporally correlated noise. We first prove the fundamental quantum limits in sensing spatially correlated Markovian noise using entangled and unentangled sensors, respectively. Focusing on power-law spatial noise correlations, which naturally arise in condensed matter systems with long-range interactions and/or near criticality, we further derive a scalable entanglement advantage when the power-law decays slowly. Then, considering a target signal with a $1/f^{p}$-type spectrum, we demonstrate that non-Markovianity may entirely modify the nature of entanglement advantage in estimating spatial noise correlations. Our protocols can be implemented using state-of-the-art quantum sensing platforms including solid-state defects, superconducting circuits, and neutral atoms.
\end{abstract}
\maketitle

Detecting spatial or temporal noise correlations provides a primary method for uncovering novel physical phenomena and probing fundamental laws of physics. For example, charge~\cite{Mailly2000} and current~\cite{Johnson1928,Nyquist1928} noise measurements can reveal emergent excitations~\cite{Levitov1996,Halperin2021} and enable a direct test of the fluctuation-dissipation theorem~\cite{Callen1951}; non-Gaussian fluctuations in cosmic microwave background may encode information about the primordial universe~\cite{Komatsu2002}; low-frequency $1/f$ noise~\cite{Clarke1987,Clarke1996,Altshuler2014,Tarucha2017} and two-level fluctuators~\cite{Schriefl2006} also directly impact dynamics of solid-state quantum systems. More recently, quantum sensors of noise have attracted immense interest due to its relevance in benchmarking quantum devices~\cite{Oliver2011,Oliver2012,Yacoby2013,Jelezko2015,Oliver2020,Nichol2022,Oliver2023} and sensing applications~\cite{Rovny2022:CovMagnetometry,Yao2023,Pasupathy2024,Ji2024:StructAnalysis, Cheng2025:MultiplxNV, Cambria2025:ScalableNV, Huxter2025:DualSpinSensor,Chong2025,Rovny2025:NVBellSensing, Zhou2025:NVBellSensing,Le2025:WidebandNVMag, Rezai2025:ProbingDynamics, Biswas2025:2DspinEnsemble}. In this context, it is natural to ask: What is the fundamental quantum limit in sensing spatiotemporally correlated noise? Moreover, are quantum resources such as entanglement necessary for achieving such optimal limits? The answer to these questions can enable the efficient detection of correlated phases in both quantum materials~\cite{Casola2018:nvCM, Rovny2024:NVforManyBody} and simulated quantum matter. Furthermore, accurate characterization of noise susceptibilities also has implications for building better many-body quantum sensors~\cite{Zhou2020:IntSpins, Bollinger2021, Colombo2022:QSNEcho, Gao2025:NVEcho, Wu2025:NVSpinSqz} and quantum information processors~\cite{Frank2005,Blatt2011,Painter2019}.

In this work, we address these questions for a paradigmatic class of fluctuations: noise with power-law spatial and/or temporal correlations. Such correlations can naturally emerge in correlated quantum many-body systems such as superfluids~\cite{Dalibard2006,Yamamoto2012}, superconductors~\cite{Clarke1987,Clarke1996}, spin glasses~\cite{Young2009,Yllanes2025}, and engineered systems such as exciton-polariton condensates~\cite{Yamamoto2012} and photonic lattices~\cite{Zubin2025}.  To solve this problem, we compute the fundamental sensitivity limits for quantum sensors of spatially or temporally correlated noise given an arbitrary initial sensor state. Our approach also incorporates arbitrary quantum control and entanglement with ancillae during a sensing protocol. Applying our result to the case of estimating spatially correlated Markovian power-law noise correlations, we further prove a scalable entanglement advantage in estimating such correlations, as long as the decay exponent is smaller than a critical value. We then extend our results to estimating non-Markovian spatial noise correlation with a power-law spectral density and find that temporal correlation in the target can drastically modify the nature of entanglement enhancement in detecting spatially correlation noise.

\begin{figure}[t]
    \centering
    \includegraphics[width=0.35\textwidth]{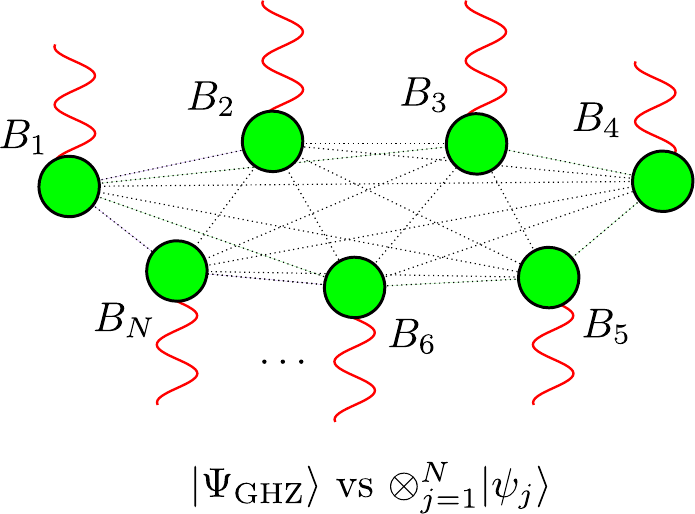}
    \caption{Schematic illustrating $N$ quantum sensors of correlated noise, equally distributed in a $1$-dimensional lattice. We are interested in estimating the overall strength of power-law noise fluctuations in space and/or time domain. To understand the existence and nature of entanglement advantage in this task, we compare the performance of entangled sensors [initialized in e.g., a GHZ-type state $\ket{\Psi_{\text{GHZ}}}$ in Eq.~\eqref{eq:ghz.def}] versus that of unentangled sensors initialized in $\otimes _{j=1}^{N} \ket{\psi _{j}}$.}
    \label{fig:schematic}
\end{figure}

{\it Problem setup.---}We consider $N$ quantum sensors of spatiotemporally correlated noise fluctuations. While our results are applicable to both classical noise as well as fluctuations generated by a quantum many-body target (see Sec.~\ref{subsec:noise.models} in the Supplemental Material for details), for concreteness, we focus on classical stationary, Gaussian noise fields as the sensing target. The target signal can be generally described by a classical stationary stochastic process $B (x,t)$ dependent on  position $x$ (which at this point can be in any number of spatial dimensions)  and time $t$. To characterize $B (x,t)$, we define its mean $\overline{B (x,t)}$ and two-point correlation function
\begin{align}
\label{eq:def.autocorr}
C (x,x'; t-t' ) \equiv \overline{\delta B (x,t) \delta B (x',t') }
,
\end{align}
where $\delta B (x,t) \equiv B (x,t) - \overline{B (x,t)} $. We use a collection of $N$ spatially distributed quantum sensors to extract properties of the noise correlation function in Eq.~\eqref{eq:def.autocorr}. We assume that $B (x,t)$ induces dephasing noise on the sensors, which is a primary mode of operation for quantum noise sensors~\cite{Rovny2024:NVforManyBody}, so that the sensor-signal interaction commutes with the intrinsic sensor Hamiltonian. As such, in the interaction picture with respect to the intrinsic sensor Hamiltonian, the sensor-signal interaction Hamiltonian reads $\hat H _{\text{int}} (t)= \sum _{j=1} ^{N} B (x _{j},t) \hat h _{j} $, where $x _{j}$ denote the sensors' spatial coordinates, and $\hat h _{j}$ are Hermitian operators describing the sensor-signal coupling. 
In a typical sensing task, $C (x,x';t-t' )$ depends on one or more parameters $\{ \xi _{\ell}\} $ through a known functional form; our goal is to estimate those parameters by measuring a collection of quantum sensors coupled to $B(x,t)$. For example, $\xi _{\ell}$ may correspond to a particular correlation matrix element or the overall strength of the noise correlation. In what follows, we will mainly focus on the case of single parameter estimation; we will briefly comment on generalization to the multiparameter case and effects of nuisance parameters but leave a systematic discussion to future work.

While each spatial eigenmode of the target noise field, as specified by the eigenvectors of $C (x,x'; t-t' ) $ in Eq.~\eqref{eq:def.autocorr}, may host distinct temporal correlation structures, for simplicity, in this work we assume the spatial and temporal components of the noise correlation factorize. 
We are interested in noise processes with long-range spatial correlations that decay as a power law: 
\begin{align}
\label{eq:Cmat.x.alpha}
C (x,x';t-t' )  \propto  |x - x' | ^{-\alpha} c_{0} (|x-x'|)
,
\end{align}
where $\alpha>0$ is a positive constant, and $c _{0} (\cdot)$ is a regularization function satisfying $\lim _{x\to 0 } x^{-\alpha} c _{0} (x)= \text{const.}$ and $c _{0} (x) =1$ for $x \ge x_{0}$. Such power-law noise correlations naturally arise in systems near criticality~\cite{Huse1988,Huse1988noneq,Young2009} and engineered optical~\cite{Zubin2025} or exciton-polariton~\cite{Yamamoto2012} system that enable long-range connectivity.
We also assume that the sensors are arranged in a one-dimensional lattice geometry with $x_j = jx_{0} $. We consider both Markovian (i.e., white) noise with no temporal correlation [$C (x,x';t-t' )  \propto \delta (t-t')$], as well as non-Markovian noise that hosts both spatial and temporal power-law correlations.

To understand the nature of entanglement advantage in this noise sensing task, we need to derive the fundamental sensitivity limits of both entangled and unentangled (i.e., separable) sensors, respectively. For this, we may consider a generic sensing protocol, where we initialize the sensors in a certain state $\hat \rho_{i}$ and let it evolve under the target noise field and quantum control, which produces a multi-sensor state
\begin{align}
\label{eq:rho.xi.def}
\hat \rho _{\xi} = \overline{\hat U (t) \hat \rho_{i} \hat U ^{\dag} (t)  }
, \quad 
\hat U (t) =\mathcal{T} e ^{-i \int ^{t}_{0} \hat H _{\text{I}} (t')  dt'},
\end{align}
where $\hat H _{\text{I}} (t) = \hat H _{\text{int}} (t)+ \hat H _{\text{c}} (t) $, and the control Hamiltonian $\hat H _{\text{c}} (t)$ does not depend on $\xi$. Note that Eq.~\eqref{eq:rho.xi.def} also applies to the case where the sensors are initially entangled with one or more ancillae and we allow general sensor-ancilla control, as long as we define $\hat \rho_{i}$ and $\hat \rho _{\xi} $ as density matrices acting on the entire sensor-ancilla system. We then perform a positive operator-valued measure (POVM)~\cite{MikeIke2011} and use the measurement statistics to build an unbiased estimator, $\check \xi$. Given a specific $\hat \rho _{\xi} $ and $M$ runs of the protocol, the optimal sensitivity, as quantified by the mean-squared error (MSE) $\Delta \check \xi$, achievable over all measurement strategies can be computed with the help of the quantum Fisher information (QFI) $\mathcal{F} _{Q} (\hat \rho _{\xi} ) $:
\begin{align}
\label{eq:qCRB}
\Delta \check \xi \ge \Delta \check \xi _{\text{opt}} = [M \mathcal{F} _{Q} (\hat \rho _{\xi} )  ] ^{-1}
,
\end{align}
where the QFI is defined as~\cite{Caves1994} 
\begin{align}
    \mathcal{F} _{Q} (\hat \rho _{\xi} ) &\equiv 8 \lim \limits_{d \xi \to 0} \frac{ 1- \sqrt{ F (\hat \rho _{\xi} , \hat \rho _{\xi + d \xi}  ) } }{d \xi ^{2}} \nonumber\\
    &= 8 \lim \limits_{d \xi \to 0} \frac{ 1- \Tr(\sqrt{\sqrt{\hat \rho _{\xi} } \hat \rho _{\xi+ d \xi}  \sqrt{\hat \rho _{\xi} }}) }{d \xi ^{2}}
\end{align}
and $F (\cdot , \cdot ) $ denotes the fidelity function. 
To quantify performance of a particular sensing strategy, we further divide $\mathcal{F} _{Q} (\hat \rho _{\xi} )$ by the signal accumulation time per shot $t_{\text{s}}$ to account for resources spent during the protocol. Supposing that the total evolution time is $T _{\mathrm{tot}}$ and the measurements and resets are sufficiently fast~\footnote{When the resets are slow, we need to optimize the QFI per shot $\mathcal{F} _{Q} (\hat \rho _{\xi} (t_{\text{s}}) )$ rather than the time-averaged QFI. In this shot-limited case, we generally expect an entanglement advantage that grows exponentially in $N$ whenever the correlation in Eq.~\eqref{eq:Cmat.x.alpha} supports a $\Theta(1)$ many-body Lindbladian gap~\cite{Wang2024:ExpAdvantage}.}, the minimal MSE can be computed via Eq.~\eqref{eq:qCRB}, as~\cite{Wang2024:ExpAdvantage}
\begin{align}
\label{eq:mse.opt.tavg.qfi}
\text{MSE} (\check \xi ) \ge \left  [T _{\mathrm{tot}} \max _{t_{\text{s}}} \frac{\mathcal{F} _{Q} (\hat \rho _{\xi} (t_{\text{s}}) )}{t_{\text{s}}} \right ] ^{-1}.
\end{align}
Note that the minimal MSE always scales as $T _{\mathrm{tot}} ^{-1} $ in estimating Markovian noise; as we will show, this scaling holds generically for estimating non-Markovian noise with power-law noise spectrum as well.

In principle, we can derive the fundamental sensitivity limits of entangled (separable) noise sensors by optimizing over all sensing strategies for sensors that involve (do not involve) entanglement.  However, such a brute-force approach would require complex optimization, which can be computationally costly. 
In what follows, we circumvent this difficulty by proving a general optimality result in sensing Markovian noise with arbitrary spatial correlations. We then discuss generalizations to non-Markovian noise.

{\it Optimality of fast-reset in detecting Markovian noise.---}We first consider the case of probing Markovian noise. Previous works have extensively studied the fundamental quantum limits of frequency estimation \textit{in the presence of} Markovian noise~\cite{Guta_2012,Kolodynski2013,Maccone2014,Kraus2014,Lukin2014QEC,Smirne2016,Maccone2016,Duer2017,Huelga2018,Lasenby2022} as well as estimation of spatially uncorrelated Markovian noise~\cite{Imai2008,Sekatski2017,Sekatski2022,DDobrzanski2023,Gefen2025}, even when arbitrary quantum control and entanglement with ancillae are allowed. In the context of noise estimation, those results make use of a key insight first pointed out in~\cite{Imai2008}, which proved that the optimal strategy (over all possible choices of the sensor's initial state, the control applied during signal accumulation, and the readout) to estimate any Pauli channel involves an \textit{independent and identically distributed (iid)}-type protocol, where one applies the channel once, immediately measures, and repeats $M$ times. 
It was later shown that this intuition applies to thermometry~\cite{Correa2015:OptThermometry,Sekatski2022},  estimating certain purely dissipative infinitesimal channels~\cite{ DDobrzanski2023, Gefen2025}, and dephasing Lindbladians with all-to-all correlations~\cite{Wang2024:ExpAdvantage}.

Here, we prove that a fast-reset strategy is optimal for estimating generic pure-dephasing dynamics. Denote by $\boldsymbol{A }(\xi) $ via $C (x_{j},x_{j'};t-t' ) = \gamma \delta (t-t') A _{jj'}(\xi) $ the sensor dephasing coefficient matrix, where $C (x_{j},x_{j'};t-t' )$ is the noise correlation in Eq.~\eqref{eq:def.autocorr}. By definition, $\boldsymbol{A }(\xi) $ is a real $\xi$-dependent positive semidefinite matrix. The sensor dynamics can be compactly described by a dephasing Lindbladian (i.e., generator of dissipative dynamics): $\mathcal{L} _{\xi} \hat \rho = \frac{\gamma}{2} \sum _{j,\ell =1} ^{n} A _{j\ell} (\xi) ( \hat h_\ell \hat\rho \hat h_j -\frac{1}{2} \{ \hat h_j  \hat h _\ell , \hat \rho\}) $.
We prove the following lower bound on the MSE of sensing $\xi$ for a general $\boldsymbol{A} (\xi)$ (c.f., Sec.~\ref{suppsec:opt.fast.reset}): 
\begin{theorem}
\label{thm:opt.ent}
(Optimality of fast-reset for entangled sensors) Given access to a dephasing Lindbladian $\mathcal{L} _{\xi}$ acting on $N$ quantum sensors for total time $T _{\mathrm{tot}}$ and arbitrary sensor initial state, control, entanglement with ancillae, and readout, the minimal MSE of an unbiased estimator $\check \xi$ reads $\mathrm{MSE} (\check \xi) _{\mathrm{opt}} = [T _{\mathrm{tot}} \max _{\hat \rho _{i}} f_Q (\hat \rho _{i})] ^{-1}  $, where $f_Q (\hat \rho _{i})= \lim _{dt \to 0} [\mathcal{F} _{Q} (e ^{\mathcal{L} _{\xi} dt} \hat \rho _{i} )/dt]$ denotes the time-averaged QFI in the asymptotic short-time limit for a sensor-ancilla system initialized in $\hat \rho _{i}$. The optimal MSE can be achieved by evolving an initial sensor-ancilla state $\hat \rho _{i}$ for an infinitesimal amount of time $dt \to 0 $ under $\mathcal{L} _{\xi}$, measure, reset, and repeat $T _{\mathrm{tot}}/dt$ times.
\end{theorem}

In order to prove the optimality of a fast-reset protocol in certain problems, previous works utilize arguments based on quantum channel estimation to derive lower bounds on the optimal MSE~\footnote{While we proved Thm.~\ref{thm:opt.ent} using a different approach, arguments based on quantum channel estimation could have been used to prove Thm.~\ref{thm:opt.ent}.}, and then explicitly show that such bounds can be achieved by a fast-reset protocol for estimating specific dissipative dynamics~\cite{Imai2008,Sekatski2017,Sekatski2022,DDobrzanski2023,Gefen2025}. While such lower bounds are very powerful as they encompass  all quantum sensing protocols with arbitrary initial states, quantum control (including entanglement with ancillae), and measurement strategies, they do not straightforwardly apply to the problem of unentangled sensors, which only involves a specific subset of initial sensor states and control. In other words, the optimality result in Theorem~\ref{thm:opt.ent} relies on the capability of choosing any sensor initial states, and does not apply to determining the optimal performance of unentangled sensors. As such, novel analytical techniques are required to address the question of entanglement advantage.

In Sec.~\ref{suppsec:opt.fast.reset} of the Supplemental Material, we prove that a fast-reset protocol also achieves the optimal sensitivity over all unentangled sensors:
\begin{theorem}
\label{thm:opt.sep}
(Optimality of fast-reset protocols for unentangled sensors) Given access to a dephasing Lindbladian $\mathcal{L} _{\xi}$ for time $T _{\mathrm{tot}}$, arbitrary separable probe state, entanglement between individual sensors and their corresponding (local) ancillae, local control, and arbitrary readout, the minimal MSE of an unbiased estimator is $\mathrm{MSE} (\check \xi) _{\mathrm{opt}} = [T _{\mathrm{tot}} \max _{\hat \rho _{i} \in \{ \hat \rho _{\mathrm{sep}} \} } f_Q (\hat \rho _{i}) ] ^{-1}  $, where $\{ \hat \rho _{\mathrm{sep}} \}$ denotes the set of all sensor states without cross-sensor entanglement. Note that the use of local ancillae and control during signal accumulation do not help reduce the MSE in this task, and $\mathrm{MSE} (\check \xi) _{\mathrm{opt}}$ only involves an optimization over all separable initial states of the sensors. Further, the optimal sensitivity can be achieved through a fast-reset protocol.
\end{theorem}
We prove this result by deriving a state-dependent bound on the rate of change in QFI that accommodates local ancillae and arbitrary local control; see Sec.~\ref{suppsec:opt.fast.reset} for details. Intuitively, a fast-reset strategy is optimal because the target signal involves Markovian noise, i.e.~has no memory: in this case, evolving the sensor under the target noise field for longer times in a single run would not enable a more efficient sensitivity than the standard quantum limit in $T _{\mathrm{tot}} $, in contrast to Hamiltonian estimation~\cite{Maccone2006,Pryde2007,Berry2009}. Note that ancillae may improve the sensor's performance in the presence of other unknown (nuisance) noise parameters~\cite{Gefen2025} or in genuinely multiparameter estimation problems.

As we show in the following section, this result allows us to rigorously compute the entanglement advantage in detecting spatially correlated Markovian noise among all possible quantum sensing protocols.

{\it Entanglement advantage in estimating spatially correlated Markovian noise.---}Entanglement advantage can be defined as the separation between the optimal performance of unentangled sensors and that of entangled ones. Using Eq.~\eqref{eq:mse.opt.tavg.qfi}, we can introduce a quantitative metric for entanglement advantage, $\mathcal{R}$, as 
\begin{align}
\mathcal{R} \equiv \frac{\max \limits_{\hat \rho _{\xi} (0)   } \max \limits_{t_{\text{s}}} \frac{\mathcal{F} _{Q} (\hat \rho _{\xi} (t_{\text{s}}) )}{t_{\text{s}}}  }{ \max \limits_{\hat \rho _{\xi} (0)  \in \{ \hat \rho _{\mathrm{sep}} \} } \max \limits_{t_{\text{s}}} \frac{\mathcal{F} _{Q} (\hat \rho _{\xi} (t_{\text{s}}) )}{t_{\text{s}}}  }
.
\end{align}
From Theorems~\ref{thm:opt.ent} and~\ref{thm:opt.sep} in the above section, we can rewrite the entanglement advantage for estimating Markovian noise in terms of the time-averaged QFI function $f_Q ( \cdot )$ in the short-time limit as $\mathcal{R} =  \max _{\hat \rho _{i}  } f_Q (\hat \rho _{i}) / \max _{\hat \rho _{i} \in \{ \hat \rho _{\mathrm{sep}} \} } f_Q (\hat \rho _{i}) $. The time-averaged QFI $f_Q (\hat \rho _{i}) $ in the short-time limit has been computed explicitly in Ref.~\cite{Brady2026:noiseQSN} for any pure initial state $\hat \rho _{i} = \ketbra{\psi _{i}} $ and can be expressed compactly as (see Sec.~\ref{suppsec:opt.fast.reset} for a detailed derivation)
\begin{align}
\label{eq:qfi.avg.max}
& f_Q (\ketbra{\psi _{i}} ) =
\nonumber \\
& 2 \gamma \sum _{ab} \left ( [\partial _{\xi}  \boldsymbol{g}  (\xi)]  [\partial _{\xi}  \boldsymbol{g} ^{\dag} (\xi)] \right ) _{ab} \braket{\psi _{i} | \Delta \hat h_a \Delta \hat h_b  |\psi _{i} } 
.
\end{align}
Here, the factor matrix $\boldsymbol{g}  (\xi)$ is defined via $\boldsymbol{A }(\xi) = \boldsymbol{g} ^{\dag}  (\xi) \boldsymbol{g}  (\xi)$, and $\Delta \hat h_j \equiv \hat h_j - \braket{\psi _{i} | \Delta \hat h_j |\psi _{i} }$. 
The expectation values $\braket{\psi _{i} | \Delta \hat h_a \Delta \hat h_b  |\psi _{i} }$ are also known as connected correlators~\cite{fetter2012quantum} and naturally arise in standard phase estimation~\cite{Helstrom1969}: the QFI in estimating a phase $\varphi$ of $e ^{-i\varphi \hat h} \ket{\psi}$ is $\propto \braket{\psi |\Delta \hat h ^{2} | \psi}$. Thus, intuitively, Eq.~\eqref{eq:qfi.avg.max} can be viewed as a weighted sum of the QFIs for standard phase estimation with generators given  by $\hat h_a $, where the weighting coefficients are set by the differentiation of the dephasing Lindbladian coefficient matrix with respect to $\xi$.

In the presence of multiple parameters $\{ \xi _{\ell} \} $, Eq.~\eqref{eq:qfi.avg.max} can be straightforwardly generalized to a QFI matrix with respect to $\{ \xi _{\ell} \} $ by extending $[\partial _{\xi}  \boldsymbol{g}  (\xi)]  [\partial _{\xi}  \boldsymbol{g} ^{\dag} (\xi)]$ to a set of products of partial derivatives of matrices $\boldsymbol{g}  (\xi)$, $[\partial _{\xi_{\ell}}  \boldsymbol{g}  (\xi)]  [\partial _{\xi_{j}}  \boldsymbol{g} ^{\dag} (\xi)]$. In this case, the matrix analog of $f_Q (\ketbra{\psi _{i}} ) $ also provides a lower bound on the MSE in estimating any linear combination of $\{ \xi _{\ell} \} $, generalizing Theorem~\ref{thm:opt.sep}~\cite{Hayashi2005,Eldredge2018,Ehrenberg2023} (see Sec.~\ref{suppsec:multiparams.qfi} for details). However, this lower bound may not be achievable in the presence of incompatible measurements and/or nuisance parameters~\cite{DemkowiczD2016,Hayashi2020}; deriving a tight lower bound in a general multiparameter noise estimation problem is an open problem under active research~\cite{Brady2026:noiseQSN, Gong2026:ScrambleSensing}.

As an illustrative example, we consider the task of estimating power-law noise correlations in Eq.~\eqref{eq:Cmat.x.alpha}, whose dissipation coefficient matrix satisfies
\begin{align}
\label{eq:Amat.x.alpha}
A _{jj'}(\xi) = \xi |j - j' | ^{-\alpha} 
, \quad 
j \ne j'
,
\end{align}
and the diagonal matrix element $A _{jj}(\xi) \propto \xi$ is a positive constant independent of $N$.
Specifically, we assume the exponent $\alpha>0$ is known, and we want to estimate the overall noise strength $\xi$. Intuitively, we expect the optimal entangled initial state of the sensors to be a GHZ-type state (up to an arbitrary relative phase shift between the two basis states)
\begin{align}
\label{eq:ghz.def}
\ket{\Psi_{\text{GHZ}}} = \frac{1}{\sqrt{2}} (\otimes _{a} \ket {h _{a,\max}} + \otimes _{a} \ket {h _{a,\min}})
\end{align}
consisting of an equal superposition between the simultaneous eigenstates of all $\hat h_a $ generators with maximal and minimal eigenvalues $h _{a,\max}$ and $h _{a,\min}$~\cite{Maccone2006}. The nature of entanglement advantage thus depends on whether an entangled state leads to a QFI that surpasses the performance of unentangled states, whose optimal time-averaged QFI scales as $\Theta ( N ) $ and can be achieved using a product sensor state $\propto \otimes _{a} ( \ket {h _{a,\max}} + \ket {h _{a,\min}}) $. For this problem, we have
\begin{theorem}
\label{thm:ent.R.white}
The quantum advantage in estimating $\xi$ in Eq.~\eqref{eq:Amat.x.alpha} satisfies 
\begin{align}
\mathcal{R} = 
\begin{cases} 
    \Theta(N^{1 - \alpha} )  & \text{for } \alpha < 1, \\ 
    \Theta ( \log N ) & \text{for } \alpha = 1, \\ 
    O (1) & \text{for } \alpha > 1,
\end{cases}
\end{align}
where we assume $N$ sensors equally distributed in real space.
\end{theorem}
See Sec.~\ref{subsec:ent.adv.white.pl} for a detailed proof. Intuitively, the optimal entangled sensors' QFI scales as $N \sum _{j=1} ^{N} j ^{-\alpha} $, as all cross-correlations in Eq.~\eqref{eq:qfi.avg.max} can contribute to the QFI for a GHZ-type initial state. In contrast, $N$ unentangled sensors always achieves a QFI of $\Theta(N)$. Taking the ratio between those two factors produces the scaling of entanglement advantage in Theorem~\ref{thm:ent.R.white}. We note that the QFI of $N$ sensors initialized in a GHZ-type state [Eq.~\eqref{eq:ghz.def}] can always be achieved by either a projective measurement given by projectors $\{\ketbra{\Psi_{\text{GHZ}}}, \mathbb{I} - \ketbra{\Psi_{\text{GHZ}}} \} $ at the end of the protocol, or alternatively, by measuring the effective Pauli-$X$ operator $\ketbra{h _{a,\max}}{h _{a,\min}} + \text{H.c.}$ on individual sensors and multiplying the oucomes.

{\it Entanglement advantage in estimating power-law spatiotemporal noise correlations.---}While Markovian noise commonly arises in environments with short correlation times and is quite ubiquitous in fluctuations of an electromagnetic origin~\cite{oqs2002book}, there are many scenarios where spatially correlated noise accompanies nontrivial temporal correlation. For example, systems near a critical point can exhibit power-law correlations in both spatial and temporal domains~\cite{Halperin1977}. This motivates us to consider the impact of temporal correlations on the nature of entanglement advantage. Specifically, as we show below, nontrivial temporal correlations can have a large impact on the existence and nature of entanglement enhancement in sensing noise correlations.

For simplicity, we again focus on the case where the spatial and temporal components of the noise correlation factorize~\footnote{The conclusion in this section can be generalized to estimating the overall noise strength of spatiotemporal correlations with more complex structure, as long as certain spatial modes host a power-law noise spectral density.}. Specifically, we are interested in power-law temporal correlations, which may naturally accompany the long-range spatial correlation in strongly correlated systems. In this case, it is convenient to transform to the Fourier domain by defining the noise spectral density function:
\begin{align}
S [ \omega ; x,x'] \equiv \int ^{+\infty } _{-\infty} C (x,x'; \tau ) e ^{i \omega \tau } d \tau
.
\end{align}
We focus on noise spectral densities with power-law spatiotemporal correlations of the following form:
\begin{align}
\label{eq:Cmat.1overf}
S [ \omega ; x,x'] = \tilde{\xi} |\omega| ^{-p} |x - x' | ^{-\alpha} c_{0} (x-x') 
,
\end{align}
with $p \in [0,3) $~\footnote{We choose the value of $p$ such that the sensor dynamics does not rely on the details of the low-frequency cutoff $c_{0}(\cdot)$; see Sec.~\ref{subsec:qfi.1qb.color} of the Supplemental Material for details.}. When $p \sim 1$, the stochastic signal is also known as $1/f^{p}$ noise and describes the dominant noise source in many solid-state and complex systems~\cite{Schriefl2006,Yacoby2013,Altshuler2014}. Similar to the Markovian case, we again consider an array of quantum sensors arranged in a lattice with $x_j = jx_{0} $, and our goal is again to estimate $ \tilde{\xi} $.

For the temporally correlated noise in Eq.~\eqref{eq:Cmat.1overf}, Theorems~\ref{thm:opt.ent} and~\ref{thm:opt.sep} no longer apply because the noise field is now heavily concentrated towards the low-frequency regime. In this case, a fast-reset protocol becomes suboptimal whenever $p>0$. For example, consider a single-qubit sensor subject to non-Markovian noise with spectral density $S [ \omega ] = S_{0} |\omega| ^{-p} $. Suppose we initialize the sensor in $\hat \rho_{i} = \ketbra{+} $ and evolve it under the target non-Markovian noise field for time $t$, with a single control $\pi$-pulse applied at $t/2$. The time-evolved sensor state $ \hat \rho (t) = \overline{\hat U (t) \hat \rho_{i} \hat U ^{\dag} (t)  }$ defined in Eq.~\eqref{eq:rho.xi.def} thus satisfies $ \braket{\uparrow \! | \hat \rho (t) | \! \downarrow } = \frac{1}{2}\exp (- \mathcal{C}_{p} S _{0} t ^{1+p} ) $~\cite{Shnirman2002,DasSarma2008,Clerk2021}, where $\mathcal{C}_{p} = \frac{2^{1-p}-1}{\pi} \Gamma(-1-p) \sin \frac{p\pi}{2}$ is a constant determined by exponent $p$. When $p>0$, the stretched-exponential decay in the sensor-qubit coherence implies that the time-averaged QFI is initially $0$ and achieves its maximum at a finite evolution time $t _{\mathrm{opt}} $. This indicates that the minimal MSE still scales as $\Theta( T _{\mathrm{tot}}  ^{-1}) $ [Eq.~\eqref{eq:mse.opt.tavg.qfi}]. Thus, a standard quantum limit persists even for estimating the strength of non-Markovian noise with a power-law noise spectral density~\footnote{We note that the quantum noise sensors' optimal performance shows a very distinct scaling from that studied in earlier works for continuously tracking a stochastic phase with a power-law spectrum~\cite{Wiseman2013}.}

More generally, we expect that a protocol evolving sensors under target noise for a finite amount of time (before measuring them) may be more efficient than a fast-reset strategy in accumulating signal, which also directly impacts the nature of entanglement advantage in this problem. This intuition is made rigorous by the following result.
\begin{theorem}
\label{thm:ent.R.1overf}
Consider $N$ qubit sensors for estimating $\tilde{\xi}$ in Eq.~\eqref{eq:Cmat.1overf} and a sensor control $\pi$-pulse sequence consisting of a set of instantaneous $\pi$-pulses applied simultaneously on all qubit sensors at $\{ \theta_{j} t_{\mathrm{s}} \} $, where $t_{\mathrm{s}}$ is the total signal accumulation time per shot, and the $j$-th pulse is applied at $\theta_{j} t_{\mathrm{s}}$. For any specific set of $\theta_{j}$, the optimal sensing protocol is to initialize the sensors in a GHZ state and choose $t_{\mathrm{s}} =t _{\mathrm{opt}} $ before detection, where the optimal evolution per shot $t _{\mathrm{opt}} $ satisfies $\tilde{\xi} ( t _{\mathrm{opt}} N ^{\frac{2-\alpha}{1+p}} ) ^{1+p}  = \mathcal{A}$, and $\mathcal{A}$ is a constant independent of $\tilde{\xi}$.  $\mathcal{A} \ne 0$ when $p>0$. The entanglement advantage now satisfies
\begin{align}
\label{eq:ent.R.color}
\mathcal{R} = 
\begin{cases} 
    \Theta(N^{\frac{1- \alpha-p}{1+p}} )  & \text{for } \alpha + p < 1, \\ 
    \Theta ( \log N ) & \text{for } \alpha + p = 1, \\ 
    O (1) & \text{for } \alpha + p > 1.
\end{cases}
\end{align}
The choice of $\{\theta_j\}$ affects the prefactor in Eq.\ \eqref{eq:ent.R.color} but not the scaling with $N$.
\end{theorem}
Note that, when $p=0$, the noise spectrum in Eq.~\eqref{eq:Cmat.1overf} becomes Markovian, and we recover Theorem~\ref{thm:ent.R.white}. However, when the noise is non-Markovian and has spectral density concentrating near the low-frequency end, the effective decay rate of a sensor prepared in a GHZ-type state [Eq.~\eqref{eq:ghz.def}] acquires an extra rescaling factor $N ^{\frac{2-\alpha}{1+p}}$. Heuristically, this is due to the fact that the sensors' evolution is mathematically equivalent to the evolution they would have under Markovian noise except for the substitution $t \to t^{1+p} $ (see Sec.~\ref{subsec:qfi.Nqb.color} for details). As a result, the optimal time-averaged QFI of entangled non-Markovian noise sensors is adjusted to $\Theta(N ^{\frac{2-\alpha}{1+p}})$ accordingly. In contrast, the unentangled sensors' characteristic dephasing rate is still $O(1)$ due to lack of entanglement, leading to an $\Theta(N)$ optimal time-averaged QFI. 
A comparison between the optimal time-averaged QFIs of entangled and unentangled sensors thus gives rise to an entanglement advantage of $ \mathcal{R}  =\Theta(N ^{\frac{2-\alpha}{1+p}-1})$ when $\alpha+p<1$, as quoted in the first line in Eq.~\eqref{eq:ent.R.color} in Theorem~\ref{thm:ent.R.1overf}. Surprisingly, this $N$-dependent rescaling in the effective decay rate of entangled sensors also removes any scalable entanglement advantage that one would have for sensing spatially correlated noise when $\alpha \in (1-p,1)$.

{\it Outlook.---}In this work, we proved the fundamental quantum limits and entanglement advantage in sensing power-law spatiotemporal noise correlations.
Our work illustrates a rich interplay between spatial and temporal noise correlation in quantum noise sensors, and how colored temporal correlation may impact their optimal performance. For future work, it would be interesting to derive the fundamental quantum limits in estimating more general spatiotemporal noise structures beyond a power-law functional dependence, e.g., spatially-correlated noise that contains sparse long-range correlations~\cite{YKLiu2021}, and correlation functions where the temporal and spatial dependences do not factorize~\cite{DeLuca2020}. One practically relevant target involves noise spectrum formed by a sum of (spatially dependent) Lorentzian peaks. It would also be exciting to extend our results to genuinely multiparameter problems: as mentioned, while a matrix version of Eq.~\eqref{eq:qfi.avg.max} readily bounds the minimal MSE in multiparameter estimation, the attainability of such bounds remains an open problem in quantum metrology. Another exciting open direction concerns function estimation~\cite{Eldredge2018,Ehrenberg2023}, where the correlation function depends on multiple unknown parameters, and we are interested in estimating one or multiple functions of those parameters.

More broadly, spatiotemporally correlated fluctuations provide a useful target in a variety of quantum metrology and learning tasks. In continuous metrology, where the readout is constrained to continuous weak measurements~\cite{Molmer2014,Molmer2016,Plenio2022,Plenio2023,Guta2023,Clerk2025,Gopalakrishnan2025,Binder2026}, it would be interesting to study the impact of entanglement on the efficiency of extracting information about spatial correlation and spectral properties of a target. Another intriguing open question concerns learning the structure of the spatiotemporal correlations of the target, e.g., extracting information about $\alpha$ and $p$ in Eq.~\eqref{eq:Cmat.1overf}. Last but not least, understanding the fundamental limits of reconstructing spatiotemporal noise correlation structures with quantum hypothesis testing~\cite{Belavkin1975,Verstraete2007,Tsang2012,Lupo2021} and quantum channel learning~\cite{Wilde2026} may also open up novel forms of entanglement advantages relevant to practical benchmarking of quantum devices.

{\it Acknowledgments.---}We thank Wenbo Sun, Sisi Zhou, Ana Asenjo-Garcia, Cosimo Rusconi, Eric Sierra, Kaixin Huang, and Yi-Kai Liu for helpful discussions.
Y.-X.W.~acknowledges support from a QuICS Hartree Postdoctoral Fellowship. 
A.J.B.\@ acknowledges support from the NRC Research Associateship Program at NIST. A.V.G.~was supported in part by  ONR MURI, DoE ASCR Quantum Testbed Pathfinder program (award No.~DE-SC0024220), NSF QLCI (award No.~OMA-2120757), NSF STAQ program, AFOSR MURI, ARL (W911NF-24-2-0107), and NQVL:QSTD:Pilot:FTL. A.V.G.~also acknowledges support from the U.S.~Department of Energy, Office of Science, National Quantum Information Science Research Centers, Quantum Systems Accelerator (award No.~DE-SCL0000121) and from the U.S.~Department of Energy, Office of Science, Accelerated Research in Quantum Computing, Fundamental Algorithmic Research toward Quantum Utility (FAR-Qu).
This material is based in part upon work supported by the U.S. Department of Energy Office of Science National Quantum Information Science Research Centers as part of the Q-NEXT center.

\clearpage

\onecolumngrid
\begin{center}
\textbf{\large Supplemental Material: Entanglement advantage in sensing power-law spatiotemporal noise correlations}
\end{center}

\setcounter{equation}{0}
\setcounter{figure}{0}
\setcounter{page}{1}
\renewcommand{\theequation}{S\arabic{equation}}
\renewcommand{\thefigure}{S\arabic{figure}}

\makeatletter
\let\addcontentsline\latex@addcontentsline
\makeatother

\tableofcontents

\subsection{Equivalence between quantum sensor dynamics generated by a pure-dephasing Markovian quantum environment and dynamics due to classical white noise}

\label{subsec:noise.models}

In this section, we show that the sensor model considered in the main text also describes evolution of a quantum sensor probing a pure-dephasing quantum environment. Specifically, we consider two classes of sensor-signal Hamiltonians, describing the sensor's interaction with a classical stochastic process $B (x,t) $ and with a quantum bath with (spatially dependent) operators $\hat B (x,t)$, respectively:
\begin{align}
\label{seq:Hint.tot}
\hat H _{\text{cl}} (t)= \sum _{j=1} ^{N} B (x _{j},t) \hat h _{j} 
, \quad 
\hat H _{\text{qu}} (t)= \sum _{j=1} ^{N} \hat B (x _{j},t) \hat h _{j}
.
\end{align}
We assume that the classical stochastic process satisfies Gaussian statistics, so that it can be fully captured by its two-point correlation function
\begin{align}
\label{seq:qubath.corr.def}
C (x,x'; t-t' ) \equiv \overline{\delta B (x,t) \delta B (x',t') }
,\quad 
\delta B (x,t) = B (x,t) - \overline{ B (x,t)}
.
\end{align}
For the quantum bath, we assume that either the bath is Gaussian (i.e., the intrinsic bath Hamiltonian is quadratic in certain bosonic raising and lowering operators, and the bath is linearly coupled to the sensor), or the bath is weakly coupled to the sensor. We note that, in considering $\hat H _{\text{qu}} (t)$, the bath can generally represent a quantum many-body system, where $\hat B (x _{j},t)$ describes a spatially dependent bath operator (e.g., polarization, excitation) coupled to the sensor. 
In this case, the environmental influence on the sensor can be fully captured by second-order correlation functions of the bath~\cite{oqs2002book}. Specifically, we introduce the (anti-)symmetrized quantum bath second-order correlation functions $C _{\mathrm{sym}} (x,x'; t-t' )$ ($C _{\mathrm{as}} (x,x'; t-t' )$), which can be defined in terms of the Heisenberg-picture bath operator $\hat B (x,t)$ and the bath (steady) state $\hat \rho _{\text{bath}}$ as 
\begin{align}
& C _{\mathrm{sym}} (x,x'; t-t' ) = \frac{1}{2}\text{Tr}  [ \{\delta \hat B (x,t), \delta \hat B (x',t') \} \hat \rho _{\text{bath}}  ]
, \quad 
& C _{\mathrm{as}} (x,x'; t-t' ) =  \text{Tr}  ( [ \hat B (x,t), \hat B (x',t') ] \hat \rho _{\text{bath}}  )
,
\end{align}
where $\delta \hat B (x,t) $ denotes the fluctuation operator of the bath noise operator:
\begin{align}
\delta \hat B (x,t) = \hat B (x,t) - \text{Tr}  [ \hat B (x,t)\hat \rho _{\text{bath}}  ]
.
\end{align}
As we show below, when the classical noise correlation matrix of $B (x,t) $ and the symmetrized correlation function of $\hat B (x,t)$ match, the sensor's dephasing dynamics generated by the two respective noise sources become identical with each other. In other words, for the purpose of analyzing the performance of the quantum sensor of noise, it suffices to focus on the corresponding (second-order) correlation function, and we could ignore other microscopic details of the noise source. In all derivations in this section, we assume the correlation function of the classical noise and the quantum bath depend on an unknown parameter $\xi$. However, we omit the $\xi$-dependence when formally writing the target's correlation function when this does not cause confusion.

We first rewrite the quantum sensor dynamics under the classical noise field, where the sensor-signal interaction is described by $\hat H _{\text{cl}} (t)$ in Eq.~\eqref{seq:Hint.tot}. In this case, the density matrix of the quantum sensor is given by (see Eq.~\eqref{eq:rho.xi.def} in the main text)
\begin{align} 
\hat \rho  = \overline{\hat U (t) \hat \rho_{i} \hat U ^{\dag} (t)  }
, \quad 
\hat U (t) =\mathcal{T} e ^{-i \int ^{t}_{0} \hat H _{\text{I}} (t')  dt'},
\end{align}
where $\hat H _{\text{I}} (t) = H _{\text{cl}} (t) + \hat H _{\text{c}} (t) $, and the control Hamiltonian $\hat H _{\text{c}} (t)$ does not depend on the unknown parameter $\xi$. When the classical stochastic process is Markovian, we can further introduce a dephasing coefficint matrix $\boldsymbol{A} $ via [c.f.~Eq.~\eqref{eq:def.autocorr} in the main text]
\begin{align} 
C (x_{j},x_{j'};t-t' ) = \gamma \delta (t-t') A _{jj'} 
.
\end{align}
In this regime, the sensor dynamics can be equivalently described by the following Lindblad master equation~\cite{oqs2002book,Frank2005}:  
\begin{align}
\label{seq:qme.clbath}
\mathcal{L} _{\xi } (t) [\hat\rho] =
- i [\hat H _{\text{c}} (t) ,  \hat\rho] 
+ \frac{\gamma}{2} \sum _{j,\ell =1} ^{n} A _{j\ell}  ( \hat h_\ell \hat\rho \hat h_j  -\frac{1}{2} \{ \hat h_j   \hat h _\ell , \hat \rho\}) 
.
\end{align}

Similarly, for a Markovian quantum mechanical environment, we define the dephasing coefficint matrix $\boldsymbol{A} $ via the bath's symmetrized correlation function as 
\begin{align}
C _{\mathrm{sym}} (x_{j},x_{j'}; t-t' ) = \gamma \delta (t-t') A _{jj'} 
.
\end{align}
The fluctuation-dissipation theorem relates the symmetrized correlation function to the antisymmetrized one, $ C _{\mathrm{sym}} (x,x'; t-t' ) $~\cite{Kubo1966}. If the bath is Markovian and in thermal equilibrium, this would require $ C _{\mathrm{as}} (x,x'; t-t' ) =0 $; for a nonequilibrium quantum environment, in principle we may have a nonzero $ C _{\mathrm{as}} (x,x'; t-t' ) \propto \delta (t-t') $, which would give rise to an anti-Hermitian contribution to the sensor's dephasing coefficient matrix. In our following discussion, we assume such contributions are negligible compared with $\boldsymbol{A} $, which holds for typical quantum Markovian environments~\cite{oqs2002book}. We can again rewrite the sensor dynamics in terms of the following Lindblad master equation~\cite{oqs2002book}, 
\begin{align}
\label{seq:qme.qubath}
\mathcal{L} _{\xi } (t) [\hat\rho] =
- i [\hat H _{\text{c}} (t) ,  \hat\rho] 
+ \frac{\gamma}{2} \sum _{j,\ell =1} ^{n} A _{j\ell} (\xi) ( \hat h_\ell \hat\rho \hat h_j  -\frac{1}{2} \{ \hat h_j   \hat h _\ell , \hat \rho\}) 
,
\end{align}
where $\hat H _{\text{c}} (t)$ denotes an additional control Hamiltonian.

Comparing Eq.~\eqref{seq:qme.clbath} with Eq.~\eqref{seq:qme.qubath}, we have shown that the sensor dynamics evolving under a quantum Markovian bath is mathematically equivalent to the dynamics under classical white noise, whenever the noise correlation function matches the bath's second-order symmetrized correlation function in Eq.~\eqref{seq:qubath.corr.def}.

\subsection{Optimality of fast-reset protocols in estimating pure-dephasing noise}

\label{suppsec:opt.fast.reset}

In this section, we provide a detailed proof of Theorems~\ref{thm:opt.ent} and~\ref{thm:opt.sep} in the main text, by deriving a tight upper bound on the time-averaged QFI achievable by quantum sensors of (spatially correlated) Markovian noise. Our proof allows entanglement with any number of ancillae and arbitrary quantum control on the joint sensor-ancilla system.

\textbf{Setup:} Consider a correlated Markovian dissipative environment whose impact on the sensor can be described by the following Lindbladian~\cite{oqs2002book,Frank2005}: 
\begin{align}
\mathcal{L} _{\xi }  \hat\rho =
\frac{\gamma}{2} \sum _{j,\ell =1} ^{n} A _{j\ell} (\xi) ( \hat L_\ell \hat\rho \hat L_j ^\dag-\frac{1}{2} \{ \hat L_j ^\dag \hat L _\ell , \hat \rho\}) 
.
\end{align}
The Lindbladian coefficient matrix $\boldsymbol{A} (\xi)$ depends on an unknown parameter $\xi \in \mathbb{R}$ whose matrix elements define the Lindbladian
\begin{align}
[\boldsymbol{A} (\xi) ] _{j\ell}  = A _{j\ell} (\xi) .
\end{align}
The matrices $\boldsymbol{A} (\xi)$ are positive semidefinite and can be complex. For convenience, we also introduce the superoperator $\partial _{\xi} \mathcal{L} _{\xi } $
\begin{align}
\partial _{\xi} \mathcal{L} _{\xi }  \hat\rho = \frac{\gamma}{2} \sum _{j,\ell =1} ^{n} 
\partial _{\xi} A _{j\ell} (\xi) ( \hat L_\ell \hat\rho \hat L_j ^\dag-\frac{1}{2} \{ \hat L_j ^\dag \hat L _\ell , \hat \rho\}) 
.
\end{align}
We can always define a matrix square root $\boldsymbol{g} (\xi) $ of the positive semidefinite matrix $\boldsymbol{A} (\xi)$ via the following equation:
\begin{align}
\boldsymbol{A} (\xi)  = \boldsymbol{g} ^{\dag} (\xi) \boldsymbol{g} (\xi) 
,
\end{align}
so that the matrix differential $\partial _{\xi} \boldsymbol{A} (\xi)$ is related to the factor matrix $\boldsymbol{g} (\xi) $ as 
\begin{align}
\label{neq:g.def.ptA}
\partial _{\xi} \boldsymbol{A} (\xi) = [\partial _{\xi} \boldsymbol{g} ^{\dag} (\xi)] \boldsymbol{g} (\xi) + \boldsymbol{g}  ^{\dag}  (\xi) \partial _{\xi} \boldsymbol{g} (\xi)
.
\end{align}

In what follows, we assume that the Lindbladian $\mathcal{L} _{\xi } $ involves pure-dephasing dynamics. More specifically, we assume that all the basis jump operators are Hermitian:
\begin{align}
\hat L_j ^\dag = \hat L_j 
,
\end{align}
and the Lindbladian coefficient matrix $\boldsymbol{A} (\xi)$ is real for all $\xi$. Specifically, when $\hat L_j = \hat h _{j} $, the Lindblad master equation describes the dynamics of a sensor under the Hamiltonian $\hat H _{\text{int}} (t)= \sum _{j=1} ^{N} B (x _{j},t) \hat h _{j} $, where the Lindbladian coefficient matrix $\boldsymbol{A} (\xi)$ is related to the noise correlation matrix $C (x,x'; t-t' ) \equiv \overline{\delta B (x,t) \delta B (x',t') }$ of the target signal via [c.f.~Eq.~\eqref{eq:def.autocorr} in the main text]
\begin{align} 
C (x_{j},x_{j'};t-t' ) = \gamma \delta (t-t') A _{jj'}(\xi) 
.
\end{align}

\textbf{Result:} Consider a pure-dephasing $\mathcal{L} _{\xi } $ with general Hermitian jump operators. For any given noise estimation protocol with an initial sensor-ancilla state $\hat \rho (t=0)$ and sensor-ancilla quantum control, which can be arbitrarily chosen, such that the sensor-ancilla system's time evolution can be described by $\hat \rho (t)$, the following upper bound on the time-averaged QFI with respect to $\xi$ holds: 
\begin{align}
\label{seq:qfi.avg.bound} 
\frac{\mathcal{F} _{Q} }{t }  \le 
& 2 \gamma \max _{ \tau } \sum _{ab} \left ( [\partial _{\xi}  \boldsymbol{g}  (\xi)]  [\partial _{\xi}  \boldsymbol{g} ^{\dag} (\xi)] \right ) _{ab}  \Tr\left ( \Delta \hat L_a ^\dag \Delta \hat L_b  \hat\rho (\tau)  \right ) 
,
\end{align}
where $t$ on the left-hand side denotes the total signal accumulation time per protocol run. Suppose the right-hand side is maximized at $\tau = \tau _{c}$, then a fast-reset protocol using $\hat \rho (\tau _{c}) $ as the initial sensor-ancilla state always achieves a time-averaged QFI that equals the upper bound on the right-hand side evaluated for the original protocol.

Note that the above result trivially includes the case without using any ancilla, when $\hat\rho (\tau)$ only act on the quantum noise sensors.

\textbf{Proof:} 
To prove Eq.~\eqref{seq:qfi.avg.bound}, we consider a general sensor-ancilla evolution parametrized by $\xi $ [see Eq.~\eqref{seq:qme.clbath}]:
\begin{align}
\mathcal{L} _{\xi,\mathrm{c} } (t) [\hat\rho] =
- i [\hat H _{\text{c}} (t) ,  \hat\rho] 
+ \frac{\gamma}{2} \sum _{j,\ell =1} ^{n} A _{j\ell} (\xi) ( \hat L_\ell \hat\rho \hat L_j ^\dag-\frac{1}{2} \{ \hat L_j ^\dag \hat L _\ell , \hat \rho\}) 
.
\end{align}
Here, $\hat H _{\text{c}} (t)$ denotes arbitrary quantum control acting on the sensor-ancilla system; $\hat H _{\text{c}} (t)$ can represent, for example, instantaneous control pulses, as long as the control does not depend on the target parameter $\xi $.
To compute the quantum Fisher information (QFI), we introduce the symmetric logarithmic derivative (SLD) $ \hat S (t)$ with respect to $\xi $ via the equation 
\begin{align}
\label{seq:symmLD.def}
\partial _{\xi } \hat\rho (t)  = \frac{1}{2} \{  \hat S (t),  \hat\rho (t)\}
,
\end{align}
which can be formally written as 
\begin{align}
\partial _{\xi } \hat\rho (t)= \frac{1}{2} \{  \hat S (t),  \hat\rho (t)\}
= \int ^{t} _{0} d \tau  \mathcal{T} e ^{\int ^{t} _{\tau } \mathcal{L} _{\xi,\mathrm{c} } (\tau _{1}) d \tau _{1} } (\partial _{\xi} \mathcal{L}_{\xi,\mathrm{c} }) \mathcal{T} e ^{\int ^{\tau} _{0 } \mathcal{L} _{\xi,\mathrm{c} } (\tau _{2}) d \tau _{2} }  \hat\rho (0)
. 
\end{align}
The QFI of state $\hat\rho (t) $ can be calculated via
\begin{align}
\label{seq:FQ.SymmLD2}
\mathcal{F} _{Q} = \Tr[\hat S (t) ^{2} \hat\rho (t) ]
,
\end{align}
whose time derivative can be derived as (noting that $\partial _{\xi} \mathcal{L}_{\xi,\mathrm{c}  } = \partial _{\xi} \mathcal{L}_{\xi }$) 
\begin{align}
\partial _{t} \mathcal{F} _{Q} & =  \Tr[(\partial _{t} \hat S (t)) \hat S (t) \hat\rho (t) + \hat S (t)  (\partial _{t} \hat S (t)) \hat\rho (t)  + \hat S (t) ^{2} \partial _{t}\hat\rho (t) ]
\\
& = \Tr[ \hat S (t) \partial _{t} (\hat S (t) \hat\rho (t) + \hat\rho (t) \hat S (t) ) - 2 \hat S (t) ^{2} \partial _{t}\hat\rho (t) +  \hat S (t) ^{2} \partial _{t}\hat\rho (t) ]
\\
& = \Tr[ 2 \hat S (t) \partial _{t} (\partial _{\xi } \hat\rho (t) ) - \hat S (t) ^{2} \partial _{t}\hat\rho (t)   ]
\\
& = 2 \Tr\left [ \hat S (t) (\partial _{\xi} \mathcal{L}_{\xi }) \mathcal{T} e ^{\int ^{t} _{0 } \mathcal{L} _{\xi,\mathrm{c} } (\tau _{2}) d \tau _{2} }  \hat\rho (0)  
+ \hat S (t) \mathcal{L} _{\xi,\mathrm{c} } (t) \int ^{t} _{0} \mathcal{T} e ^{\int ^{t} _{\tau } \mathcal{L} _{\xi,\mathrm{c} } (\tau _{1}) d \tau _{1} } (\partial _{\xi} \mathcal{L}_{\xi,\mathrm{c} }) \mathcal{T} e ^{\int ^{\tau} _{0 } \mathcal{L} _{\xi,\mathrm{c} } (\tau _{2}) d \tau _{2} }  \hat\rho (0) \right ]
\nonumber\\
& \quad - \Tr[ \hat S (t) ^{2} \mathcal{L} _{\xi,\mathrm{c} } (t) [\hat\rho (t) ]   ]
.
\end{align}
Thus, we have
\begin{align}
\partial _{t} \mathcal{F} _{Q} & =  2 \Tr\left [ \hat S (t) (\partial _{\xi} \mathcal{L}_{\xi })  \hat\rho (t)  \right ]
+ \Tr\left [  
\hat S (t) \mathcal{L} _{\xi,\mathrm{c} } (t) [\{  \hat S (t),  \hat\rho (t)\} ] \right ]
- \Tr[ \hat S (t) ^{2} \mathcal{L} _{\xi,\mathrm{c} } (t) [\hat\rho]   ]
\\
& =  2 \Tr\left [ \hat S (t) (\partial _{\xi} \mathcal{L}_{\xi })  \hat\rho (t)  \right ] 
- i  \Tr\left (  
\hat S (t) [\hat H _{\text{c}} (t) ,  \{  \hat S (t),  \hat\rho (t)\}  ] 
\right )
+ i \Tr(\hat S (t) ^{2}  [\hat H _{\text{c}} (t) ,\hat\rho]  )
\nonumber \\
& \quad +  \frac{\gamma}{2} \sum _{j,\ell =1} ^{n} \frac{1}{2}  A _{j\ell} (\xi) \Tr\left [  
([\hat L_j ^\dag, \hat S (t) ] \hat L_\ell  + \hat L_j ^\dag[\hat S (t) , \hat L_\ell  ] ) [\{  \hat S (t),  \hat\rho (t)\} ] \right ]
\nonumber \\
& \quad - \frac{\gamma}{2} \sum _{j,\ell =1} ^{n} \frac{1}{2}  A _{j\ell} (\xi) \Tr\left [  
([\hat L_j ^\dag, \hat S ^{2} (t) ] \hat L_\ell  + \hat L_j ^\dag[\hat S ^{2} (t) , \hat L_\ell  ] ) \hat\rho (t) \right ]
\\
\label{seq:qfi.dt.prod}
& =  2 \Tr\left [ \hat S (t) (\partial _{\xi} \mathcal{L}_{\xi })  \hat\rho (t)  \right ]
+ \frac{\gamma}{2} \sum _{j,\ell =1} ^{n}  A _{j\ell} (\xi) \Tr\left (  
[\hat L_j ^\dag, \hat S (t) ] [\hat L_\ell, \hat S (t) ]  \hat\rho (t) \right )
.
\end{align}
It is interesting to note that the contributions to the QFI growth rate that explicitly depend on the control Hamiltonian $\hat H _{\text{c}} (t)$ are exactly canceled in the above derivation. Intuitively, this can be understood from the fact that the control Hamiltonian $\hat H _{\text{c}} (t)$ does not have any explicit dependence on $\xi$. However, the final QFI $\mathcal{F} _{Q} $ for any given sensing protocol still depends on $\hat H _{\text{c}} (t)$ via $\hat\rho (t)$. As we will show below, Eq.~\eqref{seq:qfi.dt.prod} enables us to prove useful upper bounds on the growth rate of the QFI, which also does not have any explicit dependence on $\hat H _{\text{c}} (t)$.

We now bound the QFI via its value in the fast-reset regime. Our strategy is to relate the first term in Eq.~\eqref{seq:qfi.dt.prod} to the second term in the same equation and the fast-reset-limit QFI, which allows us to upper bound the time derivative of QFI over all control and measurement strategies. We begin by rewriting the first term in  Eq.~\eqref{seq:qfi.dt.prod} as
\begin{align}
& \Tr\left [ \hat S (t) (\partial _{\xi} \mathcal{L}_{\xi })  \hat\rho (t)  \right ] = \frac{\gamma}{2} \sum _{j,\ell =1} ^{n} \frac{1}{2}  
\partial _{\xi} A _{j\ell} (\xi) \Tr\left [  
([\hat L_j ^\dag, \hat S (t) ] \hat L_\ell  + \hat L_j ^\dag[\hat S  (t) , \hat L_\ell  ] ) \hat\rho (t) \right ]
\\
= & \frac{\gamma}{2} \sum _{j,\ell =1} ^{n} \frac{1}{2}  
\left ( [\partial _{\xi} \boldsymbol{g} ^{\dag} (\xi)] \boldsymbol{g} (\xi) + \boldsymbol{g}  ^{\dag}  (\xi) \partial _{\xi} \boldsymbol{g} (\xi) \right ) _{j\ell} \Tr\left [  
([\hat L_j ^\dag, \hat S (t) ] \hat L_\ell  + \hat L_j ^\dag[\hat S  (t) , \hat L_\ell  ] ) \hat\rho (t) \right ]
\\
= & \frac{\gamma}{2} \text{Re} \sum _{j,\ell =1} ^{n} 
\left [ \boldsymbol{g}  ^{\dag}  (\xi) \partial _{\xi} \boldsymbol{g} (\xi) \right ] _{j\ell} \Tr\left [  
([\hat L_j ^\dag, \hat S (t) ] \hat L_\ell  + \hat L_j ^\dag[\hat S  (t) , \hat L_\ell  ] ) \hat\rho (t) \right ]
. 
\end{align}
We rewrite the above expression in a form that will let us obtain a tight upper bound for the QFI. Specifically, we define $\Delta \hat L _{\ell} (t) \equiv \hat L _{j} - \Tr[ \hat L_\ell  \hat\rho (t) ] $; noting that the expectation value of any commutator between two Hermitian operators is purely imaginary, we have 
\begin{align}
& \Tr\left [ \hat S (t) (\partial _{\xi} \mathcal{L}_{\xi })  \hat\rho (t)  \right ] 
\nonumber \\
= & \frac{\gamma}{2} \text{Re} \sum _{j,\ell =1} ^{n} 
\left [ \boldsymbol{g}  ^{\dag}  (\xi) \partial _{\xi} \boldsymbol{g} (\xi) \right ] _{j\ell} \Tr\left (  
[\hat L_j ^\dag, \hat S (t) ] \Delta  \hat L_\ell  \hat\rho (t) \right )
+ \left ( [\partial _{\xi} \boldsymbol{g} ^{\dag} (\xi)] \boldsymbol{g} (\xi)  \right ) ^{*} _{\ell j } \Tr\left (  
\Delta \hat L_j ^\dag[\hat S  (t) , \hat L_\ell  ]   \hat\rho (t) \right )
.
\end{align}
By the Cauchy-Schwarz inequality, both terms on the right-hand side of the above equation can be bounded by an identical upper bound as 
\begin{align}
& \frac{\gamma}{2} \left | \sum _{j,\ell =1} ^{n} 
\left [ \boldsymbol{g}  ^{\dag}  (\xi) \partial _{\xi} \boldsymbol{g} (\xi) \right ] _{j\ell} \Tr\left (  
[\hat L_j ^\dag, \hat S (t) ] \Delta \hat L_\ell  \hat\rho (t) \right ) \right |
\\
\le & \frac{\gamma}{2} \sqrt{\sum _{j,\ell =1} ^{n}  A _{j\ell} (\xi) \Tr\left (  
[ \hat S (t) , \hat L_j ^\dag ] [\hat L_\ell, \hat S (t) ]  \hat\rho (t) \right )}
\sqrt{\sum _{ab} \left [ (\partial _{\xi} \boldsymbol{g} ^{\dag})  (\partial _{\xi} \boldsymbol{g} ) \right ]_{ab}  \Tr\left ( \Delta \hat L_a ^\dag \Delta \hat L_b  \hat\rho (t)  \right ) }
,
\end{align}
so that we can bound the following sum as 
\begin{align}
\frac{\gamma}{2} \sum _{j,\ell =1} ^{n}  A _{j\ell} (\xi) \Tr\left (  
[ \hat S (t) , \hat L_j ^\dag ] [\hat L_\ell, \hat S (t) ]  \hat\rho (t) \right )
\ge \frac{ \Tr\left [ \hat S (t) (\partial _{\xi} \mathcal{L}_{\xi })  \hat\rho (t)  \right ] ^{2}}{2 \gamma \sum _{ab} \left [ (\partial _{\xi} \boldsymbol{g} ^{\dag})  (\partial _{\xi} \boldsymbol{g} ) \right ]_{ab}  \Tr\left ( \Delta \hat L_a ^\dag \Delta \hat L_b  \hat\rho (t)  \right ) }
.
\end{align}
Substituting this back into the following equation for the time derivative of QFI:
\begin{align}
\partial _{t} \mathcal{F} _{Q} & =
2 \Tr\left [ \hat S (t) (\partial _{\xi} \mathcal{L}_{\xi })  \hat\rho (t)  \right ]
+ \frac{\gamma}{2} \sum _{j,\ell =1} ^{n}  A _{j\ell} (\xi) \Tr\left (  
[\hat L_j ^\dag, \hat S (t) ] [\hat L_\ell, \hat S (t) ]  \hat\rho (t) \right )
,
\end{align}
we obtain 
\begin{align}
& \partial _{t} \mathcal{F} _{Q} \le 
2 \Tr\left [ \hat S (t) (\partial _{\xi} \mathcal{L}_{\xi })  \hat\rho (t)  \right ]
- \frac{ \Tr\left [ \hat S (t) (\partial _{\xi} \mathcal{L}_{\xi })  \hat\rho (t)  \right ] ^{2}}{2 \gamma \max _{ \tau } \sum _{ab} \left [ (\partial _{\xi} \boldsymbol{g} ^{\dag})  (\partial _{\xi} \boldsymbol{g} ) \right ]_{ab}  \Tr\left ( \Delta \hat L_a ^\dag \Delta \hat L_b  \hat\rho (\tau)  \right ) }
.
\end{align}
Thus, we can bound the QFI growth rate as
\begin{align}
\partial _{t} \mathcal{F} _{Q} \le 
& 2 \gamma \max _{ \tau } \sum _{ab} \left [ (\partial _{\xi} \boldsymbol{g} ^{\dag})  (\partial _{\xi} \boldsymbol{g} ) \right ]_{ab}  \Tr\left ( \Delta \hat L_a ^\dag \Delta \hat L_b  \hat\rho (\tau)  \right ) 
\\
& - \frac{\left (  \Tr\left [ \hat S (t) (\partial _{\xi} \mathcal{L}_{\xi })  \hat\rho (t)  \right ] - 2 \gamma \max _{ \tau } \sum _{ab} \left [ (\partial _{\xi} \boldsymbol{g} ^{\dag})  (\partial _{\xi} \boldsymbol{g} ) \right ]_{ab}  \Tr\left ( \Delta \hat L_a ^\dag \Delta \hat L_b  \hat\rho (\tau)  \right ) \right ) ^{2}}
{2 \gamma \max _{ \tau } \sum _{ab} \left [ (\partial _{\xi} \boldsymbol{g} ^{\dag})  (\partial _{\xi} \boldsymbol{g} ) \right ]_{ab}  \Tr\left ( \Delta \hat L_a ^\dag \Delta \hat L_b  \hat\rho (\tau)  \right ) }
,
\end{align}
which leads to the following upper bound on the QFI growth rate: 
\begin{align}
\label{seq:qfi.rate.C}
\partial _{t} \mathcal{F} _{Q} &\le  2 \gamma \max _{\tau} \sum _{ab} \left [ (\partial _{\xi} \boldsymbol{g} ^{\dag})  (\partial _{\xi} \boldsymbol{g} ) \right ]_{ab}  \Tr\left ( \Delta \hat L_a ^\dag \Delta \hat L_b  \hat\rho (\tau)  \right ) \equiv \mathcal{C},
\end{align}
where $\mathcal{C}$ is a (time-independent) constant.
When we allow entanglement with ancillae and/or quantum control, we should perform an optimization with respect to the corresponding set of states in the above equation. We also note that the right-hand side may not be mathematically well-defined if the sensor involves bosonic degrees of freedom with an infinite-dimensional Hilbert space, unless we constrain the boson number or the dynamics conserve the total boson number. Under this boundedness constraint, which is trivial for finite-dimensional systems, $\mathcal{C}$ is a constant (independent of $t$).

We now use Eq.~\eqref{seq:qfi.rate.C} to prove the following upper bound on the time-averaged QFI [i.e., Eq.~\eqref{seq:qfi.avg.bound}]: 
\begin{align}
\label{seq:qfi.avg.max}
\forall t: \, \frac{\mathcal{F} _{Q} }{t }  \le 
& 2 \gamma \max _{ \tau } \sum _{ab} \left ( [\partial _{\xi}  \boldsymbol{g}  (\xi)]  [\partial _{\xi}  \boldsymbol{g} ^{\dag} (\xi)] \right ) _{ab}  \Tr\left ( \Delta \hat L_a ^\dag \Delta \hat L_b  \hat\rho (\tau)  \right ) 
,
\end{align}
where $\hat\rho (t) $ denotes the sensor state during the evolution. Given $\mathcal{F}_Q(t=0)=0$ and that, for all $\tau\in[0,t]$, the QFI growth rate satisfies
\begin{equation}
\partial_\tau \mathcal{F}_Q(\tau)\le \mathcal{C} ,
\end{equation}
where $\mathcal{C}$ is a $\tau$-independent constant,
integrating from $0$ to $t$ gives
\begin{equation}
\mathcal{F}_Q(t)-\mathcal{F}_Q(0)=\int_0^t d\tau\,\partial_\tau \mathcal{F}_Q(\tau)
\le \int_0^t d\tau\, \mathcal{C} = \mathcal{C}t.
\end{equation}
Hence
\begin{equation}
\frac{\mathcal{F}_Q(t)}{t}\le \mathcal{C},
\end{equation}
which proves~\eqref{seq:qfi.avg.max}.

The above inequality holds for arbitrary Hermitian jump operators $\hat L_a$ and an arbitrary fixed sensing protocol. We will now show that the bound becomes tight when we further optimize the time-averaged QFI over all entangled or unentangled strategies, which will replace the maximization over time with a maximization over all possible entangled or unentangled density matrices in Eq.~\eqref{seq:qfi.avg.max}: 
\begin{align}
\label{seq:qfi.avg.max.ent}
\text{entangled sensors: } & \frac{\left (\mathcal{F} _{Q} \right ) _{\text{ent}} }{t }  \le 
2 \gamma \max _{ \hat \rho } \sum _{ab} \left ( [\partial _{\xi}  \boldsymbol{g}  (\xi)]  [\partial _{\xi}  \boldsymbol{g} ^{\dag} (\xi)] \right ) _{ab}  \Tr\left ( \Delta \hat L_a ^\dag \Delta \hat L_b  \hat \rho  \right ) 
,\\
\label{seq:qfi.avg.max.sep}
\text{unentangled sensors: } & \frac{\left (\mathcal{F} _{Q} \right ) _{\text{sep}} }{t }  \le 
2 \gamma \max _{ \hat \rho _{\text{sep}}  } \sum _{ab} \left ( [\partial _{\xi}  \boldsymbol{g}  (\xi)]  [\partial _{\xi}  \boldsymbol{g} ^{\dag} (\xi)] \right ) _{ab}  \Tr\left ( \Delta \hat L_a ^\dag \Delta \hat L_b  \hat \rho _{\text{sep}}  \right ) 
.
\end{align}
To show that the bounds on the RHS of Eqs.\ \eqref{seq:qfi.avg.max.ent} and~\eqref{seq:qfi.avg.max.sep} are tight, we first define QFI growth rate at infinitesimal times as a new function $f_Q (\ketbra{\psi _{i}} )$
\begin{align}
\label{seq:qfi.avg.max.0}
& 
f_Q (\ketbra{\psi _{i}} ) \equiv \lim _{t \to 0^{+}} \partial _{t} \mathcal{F} _{Q} (e  ^{\mathcal{L} _{\xi }  t} (\ketbra{\psi _{i}}))  
.
\end{align}
Using Eqs.~\eqref{seq:symmLD.def} and~\eqref{seq:FQ.SymmLD2}, we can explicitly compute the growth rate for the QFI in the $t \to 0^{+}$ limit as
\begin{align}
\label{seq:fQ.sing}
f_Q (\ketbra{\psi _{i}} )
= 
& 2 \gamma \sum _{ab} \left ( [\partial _{\xi}  \boldsymbol{g}  (\xi)]  [\partial _{\xi}  \boldsymbol{g} ^{\dag} (\xi)] \right ) _{ab} \braket{\psi _{i} | \Delta \hat L_a ^\dag \Delta \hat L_b  |\psi _{i} } 
.
\end{align}
We have thus shown that, for the pure initial states that maximize $f_Q (\ketbra{\psi _{i}} )$ (over all entangled or unentangled states), $\mathcal{C}= \lim_{t\to 0^+}\partial_t \mathcal{F}_Q(t)$.
As such, the upper bounds in Eqs.~\eqref{seq:qfi.avg.max.ent} and~\eqref{seq:qfi.avg.max.sep} can be achieved with a fast-reset protocol using an appropriately chosen entangled or unentangled pure initial state (that maximizes $f_Q (\ketbra{\psi _{i}} )$): 
\begin{align}
\label{seq:qfi.avg.max.ent.2}
\text{entangled sensors: } & \frac{\left (\mathcal{F} _{Q} \right ) _{\text{ent}} }{t }  \le 
 \max _{ \hat \rho } f_Q ( \hat \rho )
,\\
\label{seq:qfi.avg.max.sep.2}
\text{unentangled sensors: } & \frac{\left (\mathcal{F} _{Q} \right ) _{\text{sep}} }{t }  \le 
\max _{ \hat \rho _{\text{sep}}  } f_Q ( \hat \rho _{\text{sep}} ) 
.
\end{align}
We have thus proven that the optimal time-averaged QFI $\mathcal{F} _{Q}/t$ in estimating purely dissipative dynamics with Hermitian jump operators is always achieved by a fast-reset protocol. Note that this is drastically different from the Hamiltonian estimation case, where the optimal strategy corresponds to the longest possible evolution time, which would achieve the Heisenberg limit in total evolution time.

\subsection{Extension of the upper bounds on the time-averaged QFI for noise estimation to function estimation}

\label{suppsec:multiparams.qfi}

In this section, we show that the upper bound in Eq.~\eqref{seq:qfi.avg.bound} [i.e., Eq.~\eqref{seq:qfi.avg.max}] can be straightforwardly generalized to the case of function estimation. 

Suppose that the dephasing dynamics $\mathcal{L} _{\xi } $ now depends on a set of unknown parameters $\boldsymbol{\xi} = (\xi _{1}, \xi _{2}, \ldots , \xi _{n_{\xi}})$:
\begin{align}
\mathcal{L} _{\boldsymbol{\xi} }  \hat\rho =
\frac{\gamma}{2} \sum _{j,\ell =1} ^{n} A _{j\ell} (\boldsymbol{\xi} ) ( \hat L_\ell \hat\rho \hat L_j ^\dag-\frac{1}{2} \{ \hat L_j ^\dag \hat L _\ell , \hat \rho\}) 
,
\end{align}
and we want to estimate a linear function~\cite{Eldredge2018,Ehrenberg2023} of those unknown parameters, $\boldsymbol{a} ^{T} \cdot \boldsymbol{\xi} = \sum _{\ell} a _{\ell} \xi _{\ell}$, in the presence of $n_{\xi}-1$ nuisance parameters~\cite{Hayashi2005}. Here, $a _{\ell}$ are constant weighting coefficients; this problem includes the cases where we want to estimate one of the parameters $\xi _{\ell}$. A similar derivation as the one presented in Sec.~\ref{suppsec:opt.fast.reset} lets us prove the following upper bound on the time-averaged QFI for estimating $\boldsymbol{a} ^{T} \cdot \boldsymbol{\xi}$. Denote by $\{\boldsymbol{\alpha} _{m} \} $ a set of linearly independent $n_{\xi}$ vectors satisfying $\boldsymbol{\alpha} _{1} ^T \cdot \boldsymbol{a} =1 $ and $\boldsymbol{\alpha} _{m} ^T \cdot \boldsymbol{a} =0 $ for all $m >1$. We then have
\begin{align}
\forall t: \, \frac{\mathcal{F} _{Q} }{t }  \le 
& 2 \gamma \max _{ \tau } \sum _{\ell j}
\alpha _{\ell } \alpha _{j}
\sum _{ab} \left ( [\partial _{\xi _{\ell}}  \boldsymbol{g}  (\xi)]  [\partial _{\xi_{j}}  \boldsymbol{g} ^{\dag} (\xi)] \right ) _{ab}  \Tr\left ( \Delta \hat L_a ^\dag \Delta \hat L_b  \hat\rho (\tau)  \right ) 
,
\end{align}
which generalizes Eq.~\eqref{seq:qfi.avg.max}. As the function being optimized is convex, it will take its maximum for certain pure states. Define a matrix function $\boldsymbol{f} _{\boldsymbol{Q}} (\ketbra{\psi _{i}} ) $:
\begin{align}
\left [ \boldsymbol{f} _{\boldsymbol{Q}} (\ketbra{\psi _{i}} ) \right ] _{\ell j}
\equiv\, 
& 2 \gamma \sum _{ab} \left ( [\partial _{\xi _{\ell }}  \boldsymbol{g}  (\xi)]  [\partial _{\xi _{j }}  \boldsymbol{g} ^{\dag} (\xi)] \right ) _{ab} \braket{\psi _{i} | \Delta \hat L_a ^\dag \Delta \hat L_b  |\psi _{i} }
.
\end{align}
We have thus proven the multiparameter generalization of Eqs.~\eqref{seq:qfi.avg.max.ent} and~\eqref{seq:qfi.avg.max.sep}:
\begin{align} 
\text{entangled sensors: } & \frac{\left (\mathcal{F} _{Q} \right ) _{\text{ent}} }{t }  \le 
\max _{ \ket{\psi _{i}} } \left [ \boldsymbol{\alpha} ^{T} \cdot \boldsymbol{f} _{\boldsymbol{Q}} (\ketbra{\psi _{i}} ) \cdot \boldsymbol{\alpha} \right ]
,\\ 
\text{unentangled sensors: } & \frac{\left (\mathcal{F} _{Q} \right ) _{\text{sep}} }{t }  \le 
\max _{ \ket{\psi _{i}} \in \{ \otimes _{j} \ket{\phi_{j}} \} } \left [ \boldsymbol{\alpha} ^{T} \cdot \boldsymbol{f} _{\boldsymbol{Q}} (\ketbra{\psi _{i}} ) \cdot \boldsymbol{\alpha} \right ]
.
\end{align}
Intuitively, $\{\boldsymbol{\alpha} _{m} \} $ represents a linear transformation from $\boldsymbol{a} ^{T} \cdot \boldsymbol{\xi}$ to $\boldsymbol{\xi}$. However, as we have only specified a single linear function of $\boldsymbol{\xi}$, the choice of $\{\boldsymbol{\alpha} _{m} \} $ is not unique, and different choices can produce distinct upper bounds on the QFI. Minimizing over all possible choices of $\{\boldsymbol{\alpha} _{m} \} $ subject to the constraints $\boldsymbol{\alpha} _{1} ^T \cdot \boldsymbol{a} =1 $ and $\boldsymbol{\alpha} _{m} ^T \cdot \boldsymbol{a} =0 $, we can derive a tighter bound~\cite{Eldredge2018,Ehrenberg2023}.
In the case of multiparameter noise estimation, it is generally unknown whether the above procedure can produce a tight bound. Identifying the fundamental quantum limits and optimal protocols for multiparameter estimation is an interesting open question for future investigation.

\subsection{Entanglement advantage in estimating Markovian noise with a power-law spatial noise correlation}

\label{subsec:ent.adv.white.pl}

In this section, we provide details on how to calculate the fundamental limits for sensing power-law spatially correlated Markovian noise with entangled and unentangled sensors, which proves Theorem~\ref{thm:ent.R.white}.

We begin by applying the results in Sec.~\ref{suppsec:opt.fast.reset} (Theorems~\ref{thm:opt.ent} and~\ref{thm:opt.sep} in the main text) to the following pure-dephasing Lindbladian whose coefficient matrix elements decay as a power law with respect to the sensor distances (see Eq.~\eqref{eq:Amat.x.alpha} in the main text):
\begin{align}
\label{neq:lindblad.pl}
& \mathcal{L} _{\xi }  \hat\rho =
\frac{\gamma}{2} \sum _{j,\ell =1} ^{n} A _{j\ell} (\xi) ( \hat L_\ell \hat\rho \hat L_j ^\dag-\frac{1}{2} \{ \hat L_j ^\dag \hat L _\ell , \hat \rho\}) 
, \\
\label{neq:lindblad.A.pl}
& A _{jj'}(\xi) = \xi |j - j' | ^{-\alpha} 
, \quad 
j \ne j'
,
\end{align}
where $\hat L_j ^\dag = \hat L_j $ are Hermitian operators.
Specifically, we compute the maximal time-averaged QFI for entangled and unentangled sensors in estimating $\mathcal{L} _{\xi } $, respectively, via Eqs.~\eqref{seq:qfi.avg.max.ent.2} and~\eqref{seq:qfi.avg.max.sep.2}, [c.f.~Eq.~\eqref{eq:qfi.avg.max} in the main text]:
\begin{align}
\left ( \frac{\mathcal{F} _{Q} }{t } \right ) _{\max}  = 
\max _{ \hat \rho } f_Q ( \hat \rho )
=
& 2 \gamma \max _{ \hat\rho } \sum _{ab} \left ( [\partial _{\xi}  \boldsymbol{g}  (\xi)]  [\partial _{\xi}  \boldsymbol{g} ^{\dag} (\xi)] \right ) _{ab}  \Tr\left ( \Delta \hat L_a ^\dag \Delta \hat L_b  \hat\rho   \right ) 
,
\end{align}
where $\boldsymbol{g} (\xi) $ is defined as a square-root matrix satisfying $\boldsymbol{A} (\xi)  = \boldsymbol{g} ^{\dag} (\xi) \boldsymbol{g} (\xi) $. Since our goal is to estimate an overall dephasing strength in Eq.~\eqref{neq:lindblad.pl}, so that $\boldsymbol{A} (\xi) = \xi \boldsymbol{A} (\xi=1)  $, we can prove that 
\begin{align}
[\partial _{\xi}  \boldsymbol{g}  (\xi)]  [\partial _{\xi}  \boldsymbol{g} ^{\dag} (\xi)]
= \frac{1}{4 \xi}\boldsymbol{A} (\xi) 
. 
\end{align}
Thus, we can rewrite the optimal time-averaged QFI as
\begin{align}
\label{seq:qfi.pertime.deph.rate}
\left ( \frac{\mathcal{F} _{Q} }{t } \right ) _{\max}  = 
& \frac{\gamma }{2 \xi}\max _{\hat \rho } \sum _{ab} A _{ab} (\xi)   \Tr\left ( \Delta \hat L_a ^\dag \Delta \hat L_b  \hat\rho   \right ) 
.
\end{align}
In what follows, we use the above formula to compute the time-averaged QFI for entangled and unentangled sensors, respectively.

Let us first consider unentangled strategies, which means that $\hat \rho $ in Eq.~\eqref{seq:qfi.pertime.deph.rate} stays unentangled (i.e., separable). Note that our derivation also applies to the case with any local sensor-ancilla entanglement, although (as we show below) local ancillae do not improve the final bound. In this case, we can always rewrite $\hat \rho $ as a convex sum over product states of $N$ individual sensor-ancilla subsystems:
\begin{align}
\hat \rho  = \sum _{j} p_{j} \bigotimes _{\ell=1} ^{N} \ketbra{\phi_{j,\ell} }
,
\end{align}
so that we can rewrite the time-averaged QFI as 
\begin{align}
& \left ( \frac{\mathcal{F} _{Q} }{t } \right ) _{\max}  = 
\frac{\gamma }{2 \xi}\max _{\hat \rho } \sum _{ab} A _{ab} (\xi)   \Tr\left ( \Delta \hat L_a ^\dag \Delta \hat L_b  \hat\rho   \right ) 
\\
\label{seq:qfi.pertime.prod}
= & \frac{\gamma }{2 \xi}\max _{\{ p_{j} , \otimes _{\ell=1} ^{N} \ket{\phi_{j,\ell} } \} } \sum _{j} p_{j}  
\sum _{ab} A _{ab} (\xi)  \Tr\left ( \Delta \hat L_a ^\dag \Delta \hat L_b  \ketbra{\phi_{j,a}  } \otimes \ketbra{\phi_{j,b}  }  \right ) 
.
\end{align}
As $\boldsymbol{A} (\xi)$ is a positive semidefinite coefficient matrix, the sum $ \sum _{ab} A _{ab} (\xi)   \Tr\left ( \Delta \hat L_a ^\dag \Delta \hat L_b  \hat\rho  \right ) $ in Eq.~\eqref{seq:qfi.pertime.prod} is also convex with respect to $\hat\rho  $, and we can bound the time-averaged QFI as
\begin{align}
\left ( \frac{\mathcal{F} _{Q} }{t } \right ) _{\max}  = 
& \frac{\gamma }{2 \xi}\max _{\{ p_{j} , \otimes _{\ell=1} ^{N} \ket{\phi_{j,\ell} } \} } \sum _{j} p_{j} (t) 
\sum _{ab} A _{ab} (\xi)  \Tr\left ( \Delta \hat L_a ^\dag \Delta \hat L_b  \ketbra{\phi_{j,a}  } \otimes \ketbra{\phi_{j,b}  }  \right ) 
\nonumber \\
\le  & \frac{\gamma }{2 \xi} \max _{ \{ \ket{\phi_{\ell}} \} } \sum _{a} A _{aa} (\xi) 
\braket{\phi_{a} | \Delta \hat L_a ^\dag \Delta \hat L_a   | \phi_{a}  }  
.
\end{align}
As shown in Eqs.~\eqref{seq:qfi.avg.max.0} and~\eqref{seq:fQ.sing}, the upper bound in the last line of the above equation can also be achieved by a fast-reset sensing protocol using an appropriately chosen product initial sensor state. As such, we have shown that the optimal time-averaged QFI for all separable states satisfies
\begin{align}
\label{seq:opt.qfi.sep.pl}
\left ( \frac{\mathcal{F} ^{\text{(sep)}} _{Q} }{t } \right ) _{\max}  = \gamma  \Theta (N)
= \max \limits_{\hat \rho _{i} \in \{ \hat \rho _{\mathrm{sep}} \} } f_Q (\hat \rho _{i}) 
,
\end{align}
where $f_Q (\cdot)$ is defined in Eq.~\eqref{seq:qfi.avg.max.0} and denotes the time derivative of the QFI function in the limit of infinitesimally small evolution times.

We now consider the optimal time-averaged QFI over all possible sensing strategies, including entangled protocols. We again use Eq.~\eqref{seq:qfi.pertime.deph.rate}:
\begin{align}
\left ( \frac{\mathcal{F} _{Q} }{t } \right ) _{\max}  = 
& \frac{\gamma }{2 \xi}\max _{\hat \rho } \sum _{ab} A _{ab} (\xi)   \Tr\left ( \Delta \hat L_a ^\dag \Delta \hat L_b  \hat\rho  \right ) 
.
\end{align}
By convexity of the function $\sum _{ab} A _{ab} (\xi)   \Tr\left ( \Delta \hat L_a ^\dag \Delta \hat L_b  \hat\rho   \right ) $ on the RHS in $\hat\rho $, we again have
\begin{align}
\left ( \frac{\mathcal{F} _{Q} }{t } \right ) _{\max}  = 
& \frac{\gamma }{2 \xi}\max _{\ket{\psi }} \sum _{ab} A _{ab} (\xi)   \braket{\psi | \Delta \hat L_a ^\dag \Delta \hat L_b  | \psi }
.
\end{align}
This upper bound is also achievable by using a fast-reset sensing protocol with an appropriately chosen initial sensor state $\ket{\psi _{\text{opt}}}$ that maximizes the RHS of the above equation; see below for a detailed discussion on the specific initial sensor state that achieves the optimal time-averaged QFI.
From the Cauchy-Schwarz inequality, we have 
\begin{align}
\braket{\psi | \Delta \hat L_a ^\dag \Delta \hat L_b   | \psi } \le \sqrt{ \braket{\psi | \Delta \hat L_a ^\dag \Delta \hat L_a   | \psi }
\braket{\psi | \Delta \hat L_b ^\dag \Delta \hat L_b   | \psi }}
\le \max _{\ket{\psi }} \max _{a}  \braket{\psi | \Delta \hat L_a ^\dag \Delta \hat L_a   | \psi }
,
\end{align}
so that we obtain the following upper bound on the time-averaged QFI:
\begin{align}
\left ( \frac{\mathcal{F} _{Q} }{t } \right ) _{\max} \le 
& \frac{\gamma }{2 \xi} \left ( \sum _{ab} |A _{ab} (\xi) | \right ) \max _{\ket{\psi }; \ell} \braket{\psi | \Delta \hat L_{\ell} ^\dag \Delta \hat L_{\ell}   | \psi }
.
\end{align}
From Eq.~\eqref{neq:lindblad.A.pl}, we prove the following upper bound on the time-averaged QFI:
\begin{align}
\label{seq:opt.qfi.ent.pl}
\left ( \frac{\mathcal{F} ^{\text{(ent)}} _{Q} }{t } \right ) _{\max} 
= \begin{cases} 
\gamma \Theta(N^{2 - \alpha} )  & \text{for } \alpha < 1, \\ 
\gamma \Theta ( N \log N ) & \text{for } \alpha = 1, \\ 
\gamma \Theta ( N ) & \text{for } \alpha > 1.
\end{cases}
\end{align}
It is straightforward to verify that the above bounds can be achieved by a fast-reset protocol using a GHZ-type entangled initial sensor state consisting of the simultaneous eigenstates of the Hermitian generators $\hat L_a $ with maximal and minimal eigenvalues $\ell _{a, \max/\min}$, respectively:
\begin{align}
\ket{\psi _{\text{opt}}} = \frac{1}{\sqrt{2}} \left [ (\otimes _{a} \ket {\ell _{a, \max}})  + (\otimes _{a} \ket {\ell _{a, \min}} ) 
\right ].
\end{align}

Comparing the optimal time-averaged QFI between the entangled sensors in Eq.~\eqref{seq:opt.qfi.ent.pl} and that of the separable sensors in Eq.~\eqref{seq:opt.qfi.sep.pl}, we prove the following relation for the entanglement advantage $\mathcal{R} $:
\begin{align}
\mathcal{R} = 
\begin{cases} 
    \Theta(N^{1 - \alpha} )  & \text{for } \alpha < 1, \\ 
    \Theta ( \log N ) & \text{for } \alpha = 1, \\ 
    O (1) & \text{for } \alpha > 1,
\end{cases}
\end{align}
which proves Theorem~\ref{thm:ent.R.white}.

\subsection{Entanglement advantage in estimating non-Markovian noise with power-law spatial and temporal noise correlations}

In this section, we first compute the optimal sensitivity limits of single-qubit quantum sensors, subject to a sequence of discrete control $\pi$-pulses, for estimating the strength of non-Markovian noise with a power-law spectrum. We then use this result to compute the optimal sensitivity of multiqubit entangled and unentangled sensors, respectively, for estimating non-Markovian noise with power-law spatiotemporal correlations, assuming a generic control pulse sequence. These results allow us to prove entanglement advantage in estimating such spatially correlated non-Markovian noise, which proves Theorem~\ref{thm:ent.R.1overf}.

\subsubsection{Optimal sensitivity of a single-qubit sensor for estimating non-Markovian noise with a power-law spectrum}

\label{subsec:qfi.1qb.color}

We consider a single-qubit sensor subject to non-Markovian noise represented by a classical stationary stochastic process $B (t)$, whose autocorrelation function and spectrum can be defined as
\begin{align}
C ( t-t' ) & \equiv \overline{\delta B ( t) \delta B ( t') }
, \quad 
\delta B (t) \equiv B (t) - \overline{B (t)}
, \\
S _{BB}[\omega] & \equiv \int ^{+\infty } _{-\infty} C (\tau ) e ^{i \omega \tau } d \tau
.
\end{align}
We focus on the case where the spectral density of the target noise field follows a power-law frequency dependence:
\begin{align}
\label{seq:nsd.1overf}
S _{BB}[\omega] = \xi |\omega | ^{-p} c_{1} (|\omega | / \omega_{\text{cut}} ) ,
\end{align}
where we assume that the exponent $p$ is known, and $c_{1} (|\omega |)$ denotes a regularization function in frequency domain such that $\lim _{x\to 0 } x^{-p} c _{1} (x)= \text{const.}$ and $c _{1} (x) =1$ for $|x|> 1$. Our goal is to estimate $ \xi$. Note that, in the absence of the regularization function $c_{1} (\cdot )$, the noise spectral density function defined in Eq.~\eqref{seq:nsd.1overf} would diverge in the $|\omega| \to 0$ limit. As such, physically, we would expect an infrared cutoff $\omega_{\text{cut}}$ to apply at sufficiently low frequencies. In practice, however, sometimes the power-law behavior in Eq.~\eqref{seq:nsd.1overf} may hold in a frequency range spanning multiple orders of magnitude, and the anticipated IR frequency cutoff may not emerge in the regime of frequencies accessible in experiments~\cite{Oliver2011,Yacoby2013}. In this section, we are interested in the parameter regime where the sensor dynamics is agnostic to the details of such low-frequency cutoff, which would impose certain constraints on the value of $p$ (see Eq.~\eqref{seq:zeta.pl}).

The sensor dynamics is described by a single-qubit Hamiltonian
\begin{align}
\hat H _{\text{sing}} (t) = \frac{1}{2} B ( t) \hat Z + \hat H _{\text{c}} (t)
,
\end{align}
where the control Hamiltonian $\hat H _{\text{c}} (t)$ consists of a sequence of instantaneous $\pi$ pulses (also known as echo pulses) applied at times $\theta_{j} t_{\mathrm{s}}$ for $j=1,2,\ldots , K$, and $t_{\mathrm{s}}$ is the total sensor evolution time per run of the experiment. Note that the control considered in this section is more restrictive compared to Sec.~\ref{suppsec:opt.fast.reset}, where arbitrary quantum control is allowed.

Following standard treatment in qubit noise spectroscopy, we transform to a toggling frame defined with respect to the control Hamiltonian $\hat H _{\text{c}} (t)$, where the control pulses effectively modulate the qubit-signal coupling according to a time-dependent function $F(t)$. Specifically, the modulation function $F(t)$ satisfies $|F(t) | =1$ and is initialized to $F(t=0) = 1$ at the beginning of the time evolution, and switches sign every time a control $\pi$ pulse (either along $x$ or $y$ direction) is applied. The sensor density matrix by the end of time evolution (in each run) is given by
\begin{align}
\hat \rho (t_{\mathrm{s}} ) = \overline{\hat U _{\text{I}} (t) \hat \rho_{i} \hat U _{\text{I}}  ^{\dag} (t)  }
, \quad 
\hat U _{\text{I}}  (t_{\mathrm{s}}) =\mathcal{T} e ^{-i \int ^{t_{\mathrm{s}}}_{0} F (t') \hat B(t')  dt'}
,
\end{align}
so that the qubit coherence function at $t=t_{\mathrm{s}}$ can be computed as 
\begin{align}
\braket{\uparrow \! | \hat \rho (t_{\mathrm{s}}) | \! \downarrow }
= \braket{\uparrow \! | \hat \rho (t=0) | \! \downarrow }
\cdot \exp \left ( -\frac{1}{2} \, \overline{ \left [ \int ^{t_{\mathrm{s}}}_{0} F (t') \hat B(t')  dt'\right ] ^{2}} \right )
.
\end{align}
We can rewrite the above equation in terms of the noise spectrum $S _{BB}[\omega]$ of the classical noise field $B(t)$ and the Fourier transform of the modulation function $F [\omega ] \equiv \int ^{t_{\mathrm{s}}}_{0} F (t') e ^{i\omega t'} dt'$:
\begin{align}
\frac{\braket{\uparrow \! | \hat \rho (t_{\mathrm{s}}) | \! \downarrow }}
{\braket{\uparrow \! | \hat \rho (t=0) | \! \downarrow }}
= \exp \left ( -\frac{1}{4\pi} \int ^{+\infty } _{-\infty} |F [\omega ]| ^{2} S _{BB}[\omega] d \omega \right )
\equiv  e ^{-\zeta (t_{\mathrm{s}}) }
,
\end{align}
where we introduce the dephasing factor function $\zeta (t_{\mathrm{s}})$ for convenience. By definition,
\begin{align}
\zeta (t_{\mathrm{s}}) = 
\frac{1}{4\pi} \int ^{+\infty } _{-\infty} |F [\omega ]| ^{2} S _{BB}[\omega] d \omega
.
\end{align}
In the following discussion, we first compute $F [\omega ]$ for different control sequences, and then use that result to derive the optimal time-averaged QFI of the corresponding sensing pulse sequence.

The Fourier-domain function $F [\omega ]$ is also known as the filter function of the control pulse sequence. In the absence of any control pulse, $F(t) =1$, and the filter function reads
\begin{align}
F _{\text{FID}} [\omega ] = \int ^{t_{\mathrm{s}}}_{0} F (t') e ^{i\omega t'} dt'
= \frac{e ^{i\omega t_{\mathrm{s}}}-1}{i\omega }
.
\end{align}
For a given control pulse sequence specified by the sequence of time points $\theta_{j} t_{\mathrm{s}}$ ($j=1,2,\ldots , K$), we also explicitly compute its filter function as 
\begin{align}
F _{\{\theta_{j}\} } [\omega ] = \int ^{t_{\mathrm{s}}}_{0} F (t') e ^{i\omega t'} dt'
= \frac{\sum _{\ell =1} ^{K+1} (-)^{\ell-1} (e^{i \omega \theta_{j} t_{\mathrm{s}}} - e^{i \omega \theta_{j-1} t_{\mathrm{s}}}) }{i \omega }
= \frac{2 \sum _{\ell =1} ^{K} (-)^{\ell-1} e^{i \omega \theta_{j} t_{\mathrm{s}}} -1 + (-) ^{K} e ^{i \omega t_{\mathrm{s}} } }{i \omega }
,
\end{align}
where we set $\theta_{j=0}=0 $ and $\theta_{j=K+1}=1 $. The dephasing factor function now becomes
\begin{align}
\text{No pulse: }
\zeta _{\text{FID}}  (t_{\mathrm{s}})&  = 
\frac{1}{4\pi} \int ^{+\infty } _{-\infty} |F _{\text{FID}}  [\omega ]| ^{2} S _{BB}[\omega] d \omega = \frac{\xi}{\pi} \int ^{+\infty } _{-\infty} \frac{\sin ^{2} \frac{\omega t_{\mathrm{s}}}{2}}{\omega ^{2+p} } c_{1} \left(\frac{|\omega |}{\omega _{\text{cut}}} \right) d \omega
\\
& = \frac{\xi}{\pi} t_{\mathrm{s}} ^{1+p} \int ^{+\infty } _{-\infty} c_{1} \left(\frac{| y |}{\omega _{\text{cut}}t_{\mathrm{s}}  } \right) \sin ^{2} \frac{y}{2} \frac{dy}{y ^{2+p}}, \\
\text{With pulse(s): }
\zeta _{\{\theta_{j}\} }   (t_{\mathrm{s}})&  = \frac{1}{4\pi} \int ^{+\infty } _{-\infty} |F _{\{\theta_{j}\} }  [\omega ]| ^{2} S _{BB}[\omega] c_{1} \left(|\omega | / \omega_{\text{cut}} \right) d \omega 
\\
& = \frac{\xi}{4\pi} t_{\mathrm{s}} ^{1+p} \int ^{+\infty } _{-\infty} c_{1} \left(\frac{| y |}{\omega _{\text{cut}}t_{\mathrm{s}}  } \right) 
\left |2 \sum _{\ell =1} ^{K} (-)^{\ell-1} e^{i \theta_{j} y } -1 + (-) ^{K} e ^{i y } \right | ^{2} \frac{dy}{y ^{2+p}}
.
\end{align}
For generic values of $p$, the integrals in the equations above would depend on details of the regularization function $c_{1} (\cdot)$. We are interested in the range of $p$ where the sensor qubit dynamics is agnostic to the details of the cutoff: specifically, we can show
\begin{subequations}
\label{seq:zeta.pl}
\begin{align}
\text{No pulse: if } p \in (-1,1),& \quad 
\zeta _{\text{FID}}  (t_{\mathrm{s}})  =  \xi t_{\mathrm{s}} ^{1+p} \mathcal{C} _{\text{FID}}  
, \\
\text{With pulse(s): if } p \in (-1,3),& \quad 
\zeta _{\{\theta_{j}\} }   (t_{\mathrm{s}})  =  \xi t_{\mathrm{s}} ^{1+p} \mathcal{C} _{\{\theta_{j}\} }
,
\end{align}   
\end{subequations}
where $\mathcal{C} _{\text{FID}}  $ and $\mathcal{C} _{\{\theta_{j}\} }$ are constant coefficients set by the control (see, e.g.,~\cite{Clerk2021}): 
\begin{align}
& \mathcal{C} _{\text{FID}}  = \frac{\Gamma(-1-p) }{\pi} \sin \frac{p\pi}{2}
, \\
\label{seq:calC.gen.echo}
& \mathcal{C} _{\{\theta_{j}\} } = \frac{\Gamma(-1-p) }{\pi} \left \{ -4 \sum _{\ell < \ell '} ^{K} (-) ^{\ell + \ell '} (\theta_{\ell} - \theta_{\ell'}) ^{1+p} + 2 \sum _{\ell =1} ^{K} (-) ^{\ell} \left [ (-) ^{K} (1-\theta_{\ell}) ^{1+p} - \theta_{\ell}^{1+p} \right ] + (-)^{K}\right \} \sin \frac{p\pi}{2}
.
\end{align}
Here, $\Gamma (\cdot)$ denotes the gamma function.
Particularly, for a single $\pi$ pulse applied at $\theta_{j=1} = \frac{1}{2}$, also known as a Hahn-echo sequence, we have
\begin{align}
\mathcal{C} _{\text{Hahn} } =  \frac{2 ^{1-p} -1}{\pi} \Gamma(-1-p) \sin \frac{p\pi}{2}
.
\end{align}
Note that this coefficient function at $p=1$ and $p=2$ contains a diverging function $\Gamma(-1-p)$. In those cases, the coefficient function can be defined via a smooth extension by taking the $p\to 1$ and $p\to 2$ limits, respectively.

The dephasing factor function in Eq.~\eqref{seq:zeta.pl} fully specifies the sensor dynamics under the non-Markovian noise with the power-law spectrum in Eq.~\eqref{seq:nsd.1overf} and the control pulse sequence corresponding occuring at times 
$\{\theta_{j} t_{\mathrm{s}} \} $. For any specific control pulse sequence, we can compute the QFI of the final sensor state $\hat \rho (t_{\mathrm{s}} | \{\theta_{j}\} ) $ as
\begin{align}
\mathcal{F } _{Q} (\hat \rho (t_{\mathrm{s}} )| \{\theta_{j}\}  ) = \frac{(t_{\mathrm{s}} ^{1+p} \mathcal{C} _{\{\theta_{j}\} }) ^{2} }{e ^{2 \xi t_{\mathrm{s}} ^{1+p} \mathcal{C} _{\{\theta_{j}\} } } -1}.
\end{align}
Maximizing over all possible total evolution times $t_{\mathrm{s}} $, the optimal time-averaged QFI is given by 
\begin{align}
\label{seq:qfiovert.rescale}
\max _{t_{\mathrm{s}} } \frac{\mathcal{F } _{Q} (\hat \rho (t_{\mathrm{s}} )| \{\theta_{j}\}  )}{t_{\mathrm{s}} }
= \max _{t_{\mathrm{s}} } \frac{t_{\mathrm{s}} ^{1+2p} \mathcal{C} _{\{\theta_{j}\} } ^{2} }{e ^{2 \xi t_{\mathrm{s}} ^{1+p} \mathcal{C} _{\{\theta_{j}\} } } -1}
= \xi  ^{-\frac{1+2p}{1+p}} \mathcal{C} _{\{\theta_{j}\} } ^{\frac{1}{1+p}}
\max _{y} g(y)
, \quad 
g(y) \equiv \frac{y ^{\frac{1+2p}{1+p}}}{e ^{2y} -1}
,
\end{align}
and the maximum is achieved when $y _{0} = \xi t_{\mathrm{s}} ^{1+p} \mathcal{C} _{\{\theta_{j}\} }$ is taken to be the variable that maximizes the function $ g(y) $. Generally for any $p>0$ and $y\in [0,+\infty)$, the function $ g(y) $ is monotonically increasing until reaching its global maximum at $y_{0}$, and then becomes monotonically decreasing for all $y >y_{0}$; Fig.~\ref{suppfig:fyoverp} depicts the function for a few values of $p$.
We see that, when $p>0$, the function $ g(y)  $ reaches its maximum at a nonzero value $y _{0} = \xi t_{\mathrm{s}} ^{1+p} \mathcal{C} _{\{\theta_{j}\} }>0$, which means that the time-averaged QFI for estimating a power-law noise spectrum is maximized at a finite evolution time per run (for any specific control echo pulse sequence). This stands in contrast to the case of estimating Markovian noise, where the time-averaged QFI is always achieved by a fast-reset protocol.

\begin{figure}[h]
    \centering
    \includegraphics[width=0.5\textwidth]{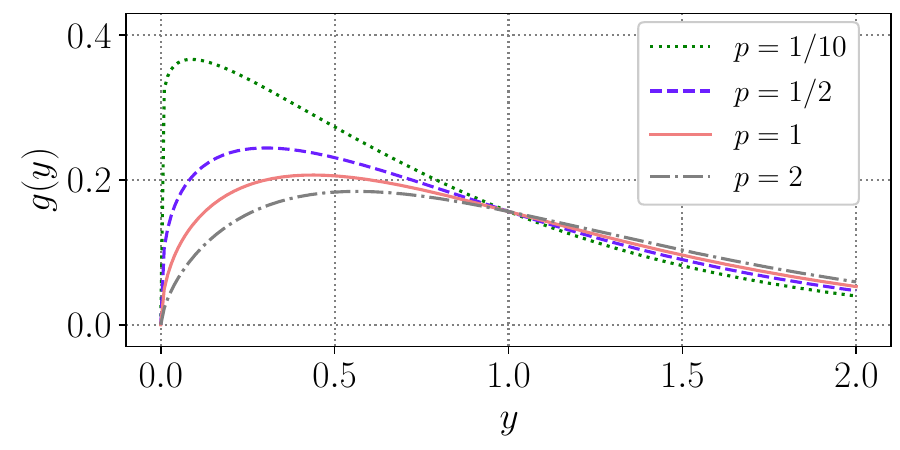}
    \caption{Plot of function $ g(y) =  {y ^{\frac{1+2p}{1+p}}} / (e ^{2y} -1)$ [Eq.~\eqref{seq:qfiovert.rescale}] versus $y$ for $p=1/10$, $1/2$, $1$, and $2$. As shown in the plot, the function reaches its maximum at a finite $y _{0} \in (0,1)$.}
    \label{suppfig:fyoverp}
\end{figure}

\subsubsection{Entanglement advantage in estimating non-Markovian noise with power-law spatiotemporal correlations}
\label{subsec:qfi.Nqb.color}

We now compute the optimal performance of multiqubit sensors in estimating non-Markovian noise with factorized power-law spatiotemporal correlations. Specifically, we consider the multiqubit sensor described in the main text, where the sensor-signal interaction Hamiltonian reads $\hat H _{\text{int}} (t)= \sum _{j=1} ^{N} B (x _{j},t) \hat h _{j} $. $B (x ,t)$ describes a spatiotemporally correlated noise field with correlation functions satisfying (see Eq.~\eqref{eq:Cmat.1overf} in the main text)
\begin{align}
&C (x,x'; t-t' ) \equiv \overline{\delta B (x,t) \delta B (x',t') },
\\
&S [ \omega ; x,x'] \equiv \int ^{+\infty } _{-\infty} C (x,x'; \tau ) e ^{i \omega \tau } d \tau
= \tilde{\xi} |\omega| ^{-p} |x - x' | ^{-\alpha} c_{0} (x-x') 
.
\end{align}
Because the temporal and spatial noise correlations factorize, from Eqs.~\eqref{seq:zeta.pl} in the above section, we can relate the density matrix $\hat \rho (t_{\mathrm{s}} )$ of a multiqubit sensor evolving under the \textit{non-Markovian} noise at time $t_{\mathrm{s}} $ to the density matrix $\hat \rho _{\text{eff}} (t_{\mathrm{eff}} )$ of a fictitious multiqubit sensor undergoing \textit{Markovian} noise with correlation function $C _{\mathrm{eff}} (x,x'; t-t' ) = \gamma \delta (t-t') |x - x' | ^{-\alpha} c_{0} (x-x')  $ at an appropriately chosen  effective time, as
\begin{align}
\hat \rho (t_{\mathrm{s}} ) = \hat \rho _{\text{eff}} (t_{\mathrm{eff}} = \gamma ^{-1} t_{\mathrm{s}} ^{1+p} \mathcal{C} _{\{\theta_{j}\} } )
, 
\end{align}
where $\mathcal{C} _{\{\theta_{j}\} }$ is again a constant coefficient that depends on the control pulse sequence (see Eq.~\eqref{seq:calC.gen.echo} for an analytical form). This mapping between the physical sensor state under non-Markovian noise and a fictitious sensor state evolving for an effective time (albeit unphysical) makes it convenient for us to prove general bounds on the performance of non-Markovian noise sensors using the general upper bound on the QFI growth rate in Eq.~\eqref{seq:qfi.rate.C} proven in Sec.~\ref{suppsec:opt.fast.reset}.

More concretely, we bound the time-averaged QFI of the original sensor evolving under non-Markovian noise via
\begin{align}
& \frac{\mathcal{F } _{Q} (\hat \rho (t_{\mathrm{s}} )| \{\theta_{j}\}  )}{t_{\mathrm{s}}}
= \frac{\mathcal{F } _{Q} (\hat \rho _{\text{eff}} (t_{\mathrm{eff}} = \gamma ^{-1} t_{\mathrm{s}} ^{1+p} \mathcal{C} _{\{\theta_{j}\} } )  )}{t_{\mathrm{s}}}
\\
= & \frac{\mathcal{F } _{Q} (\hat \rho _{\text{eff}} (t_{\mathrm{eff}} = \gamma ^{-1} t_{\mathrm{s}} ^{1+p} \mathcal{C} _{\{\theta_{j}\} } )  )}{\left ( \frac{\gamma t_{\mathrm{eff}} }{\mathcal{C} _{\{\theta_{j}\} }} \right ) ^{\frac{1}{1+p}}}
\\
\label{seq:quasi.tavg.qfi}
= & \left ( \frac{\mathcal{C} _{\{\theta_{j}\} }}{\gamma  } \right ) ^{\frac{1}{1+p}} \frac{\mathcal{F } _{Q} (\hat \rho _{\text{eff}} (t_{\mathrm{eff}}   )  )}{t_{\mathrm{eff}} ^{\frac{1}{1+p}}}
.
\end{align}
Thus, the evaluation of the optimal time-averaged QFI for the \textit{non-Markovian} noise sensor can be mapped to the following optimization problem on a quasi-time-averaged QFI of a \textit{Markovian} noise sensor:
\begin{align}
\max_{t_{\mathrm{s}} } \frac{\mathcal{F } _{Q} (\hat \rho (t_{\mathrm{s}} )| \{\theta_{j}\}  )}{t_{\mathrm{s}}}
= \left ( \frac{\mathcal{C} _{\{\theta_{j}\} }}{\gamma  } \right ) ^{\frac{1}{1+p}} 
\max _{t_{\mathrm{eff}}  } \frac{\mathcal{F } _{Q} (\hat \rho _{\text{eff}} (t_{\mathrm{eff}}   )  )}{t_{\mathrm{eff}} ^{\frac{1}{1+p}}}
.
\end{align}
In particular, as the mapping only modifies the time dependence of the sensor's density matrix and does not impact its entanglement properties, we can relate the entangled and unentangled non-Markovian noise sensors to their counterparts undergoing the mapping:
\begin{align}
\label{seq:qfi.max.color.teff}
\max_{t_{\mathrm{s}} } \frac{\mathcal{F } _{Q} (\hat \rho _{\text{sep}} (t_{\mathrm{s}} )| \{\theta_{j}\}  )}{t_{\mathrm{s}}}
= \left ( \frac{\mathcal{C} _{\{\theta_{j}\} }}{\gamma  } \right ) ^{\frac{1}{1+p}} 
\max _{t_{\mathrm{eff}}  } \frac{\mathcal{F } _{Q} (\hat \rho _{\text{eff;sep}} (t_{\mathrm{eff}}   )  )}{t_{\mathrm{eff}} ^{\frac{1}{1+p}}}
.
\end{align}

As $\hat \rho _{\text{eff}} (t_{\mathrm{eff}}   ) $ describes the quantum state of a sensor evolving under Markovian spatially correlated noise, it satisfies Eq.~\eqref{seq:qfi.rate.C}; combining the upper bound in Eq.~\eqref{seq:qfi.rate.C} and convexity of separable states, we can prove the following scaling of the optimal time-averaged QFI of the original non-Markovian sensors, in the absence of any cross-sensor entanglement:
\begin{align}
\label{seq:opt.qfiavg.color.sep}
\max_{t_{\mathrm{s}} } \frac{\mathcal{F } _{Q} (\hat \rho _{\text{sep}} (t_{\mathrm{s}} )| \{\theta_{j}\}  )}{t_{\mathrm{s}}}
= \Theta (N)
.
\end{align}
An example unentangled sensor initial state that achieves this optimal scaling consists of the product state $\otimes _{a} \frac{\ket {h _{a,\max}} + \ket {h _{a,\min}}}{\sqrt{2}} $, where $\ket {h _{a,\max}}$ and $\ket {h _{a,\min}}$ denotes the eigenstates of $\hat h _{a} $ with maximal and minimal eigenvalues, respectively.
As for the entangled sensors, when $\alpha+p<1$, we can again use Eq.~\eqref{seq:qfi.rate.C} and convexity of quantum states to prove that the maximal quasi-time-averaged QFI in Eq.~\eqref{seq:quasi.tavg.qfi} is achieved by an entangled GHZ-type sensor initial state $ (\otimes _{a}\ket {h _{a,\max}} + \ket {\otimes _{a}h _{a,\min}}) / \sqrt{2} $. Further, combining Eq.~\eqref{seq:opt.qfi.ent.pl} and Eq.~\eqref{seq:zeta.pl}, we can derive the dephasing factor function for the GHZ-type sensor state as
\begin{align}
p \in (-1,3):& \quad 
\zeta _{\{\theta_{j}\} }   (t_{\mathrm{s}})  =  \xi t_{\mathrm{s}} ^{1+p} \mathcal{C} _{\{\theta_{j}\} }
\times \begin{cases} 
\Theta(N^{2 - \alpha} )  & \text{for } \alpha < 1, \\ 
\Theta ( N \log N ) & \text{for } \alpha = 1, \\ 
\Theta ( N ) & \text{for } \alpha > 1,
\end{cases}
\end{align}
which effectively adds an $N$-dependent factor to the constant coefficient $\mathcal{C} _{\{\theta_{j}\} }$ in the dephasing factor. From Eq.~\eqref{seq:qfi.max.color.teff}, we can derive the optimal time-averaged QFI for the non-Markovian noise sensor as 
\begin{align}
\max_{t_{\mathrm{s}} } \frac{\mathcal{F } _{Q} (\hat \rho  (t_{\mathrm{s}} )| \{\theta_{j}\}  )}{t_{\mathrm{s}}}
= \begin{cases} 
\Theta(N^{\frac{2 - \alpha}{1+p}} )  & \text{for } \alpha +p < 1, \\ 
\Theta ( N \log N ) & \text{for } \alpha+p = 1, \\ 
\Theta ( N ) & \text{for } \alpha +p > 1.
\end{cases}
\end{align}
In particular, when $\alpha +p > 1$, the optimal strategy involves using a product state, whereas the GHZ-type state $ (\otimes _{a}\ket {h _{a,\max}} + \ket {\otimes _{a}h _{a,\min}}) / \sqrt{2} $ only achieves a time-averaged QFI of $\Theta(N^{\frac{\max \{ 1, 2-\alpha\}}{1+p}})$.
Dividing the above equation by the optimal time-averaged QFI for the unentangled sensors in Eq.~\eqref{seq:opt.qfiavg.color.sep}, we obtain the following entanglement advantage for the non-Markovian noise sensors:
\begin{align}
\mathcal{R} = 
\begin{cases} 
    \Theta(N^{\frac{1- \alpha-p}{1+p}} )  & \text{for } \alpha + p < 1, \\ 
    \Theta ( \log N ) & \text{for } \alpha + p = 1, \\ 
    O (1) & \text{for } \alpha + p > 1,
\end{cases}
\end{align}
which proves Theorem~\ref{thm:ent.R.1overf}.

\end{document}